\documentclass[11pt]{article}
\RequirePackage{amsthm,amsmath,amsfonts,amssymb}
\RequirePackage[authoryear]{natbib}
\RequirePackage[colorlinks,citecolor=blue,urlcolor=blue]{hyperref}
\usepackage[utf8]{inputenc}
\RequirePackage{graphicx}

\usepackage[margin=1in]{geometry}

\usepackage{amsmath}
\usepackage{comment}
\usepackage{appendix}
\usepackage{amsfonts}
\usepackage{amssymb}
\usepackage{amsxtra}
\usepackage{amsthm}
\usepackage{bbm}
\usepackage{float}
\usepackage{color}

\newtheorem{theorem}{Theorem}[section]

\newtheorem{proposition}[theorem]{Proposition}
\newtheorem{corollary}[theorem]{Corollary}

\title{\textsc{Group-weighted conformal prediction}}

\author{
  \textbf{Aabesh Bhattacharyya, Rina Foygel Barber} \\
  Department of Statistics,\\
  University of Chicago,\\
  Chicago,IL\\
  \texttt{\{aabesh,rina\}@uchicago.edu} \\
}

\date{}

\begin{document}
\allowdisplaybreaks
\maketitle

\begin{abstract}
Conformal prediction (CP) is a method for constructing a prediction interval around the output of a fitted model, whose validity does not rely on the model being correct---the CP interval offers a coverage guarantee that is distribution-free, but relies on the training data being drawn from the same distribution as the test data. A recent variant, weighted conformal prediction (WCP), reweights the method to allow for covariate shift between the training and test distributions. However, WCP requires knowledge of the nature of the covariate shift---specifically, the likelihood ratio between the test and training covariate distributions. In practice, since this likelihood ratio is estimated rather than known exactly, the coverage guarantee may degrade due to the estimation error. 
In this paper, we consider a special scenario where observations belong to a finite number of groups, and these groups determine the covariate shift between the training and test distributions---for instance, this may arise if the training set is collected via stratified sampling. Our results demonstrate that in this special case, the predictive coverage guarantees of WCP can be drastically improved beyond the bounds given by existing estimation error bounds.
\end{abstract}

\section{Introduction}\label{sec:intro}
The conformal prediction (CP) framework offers a method for constructing prediction intervals, for the problem of predicting a response variable $Y\in\mathcal{Y}$ 
based on features $X\in\mathcal{X}$, in a distribution-free way (see   \citep{vovk2005algorithmic} for background). 
In its standard form, conformal prediction deals with the case of i.i.d.\ data (or, more generally, exchangeable data)---namely, the training data used
for constructing the prediction intervals, and the test data on which the prediction intervals are deployed, are assumed to be drawn i.i.d.\ from
the \emph{same} distribution. 

To relax this assumption, \citep{tibshirani2019conformal} extend conformal prediction
to the setting of covariate shift, where the test data distribution $Q$ may differ from the training data distribution $P$ in terms of the marginal
distribution of the covariates---that is, $Q_X$ may differ from $P_X$ (but the dependence of the response $Y$ given features $X$ is assumed
to be the same across training and test data, i.e., the conditional distributions $P_{Y|X}$ and $Q_{Y|X}$ are assumed to be equal).
However, the resulting method, weighted conformal prediction (WCP), requires knowledge of the covariate shift (specifically, the ratio between the marginal
distributions $P_X$ and $Q_X$, used to \emph{reweight} the training data to better resemble the test data distribution). For example, if $P_X$ and $Q_X$ have densities $f$ and $g$, respectively,
implementing WCP requires access to a weight function $w(x) \propto g(x)/f(x)$. However, in a real data setting, it is not realistic 
to assume that $w(x)$ is known exactly. In practice the weight function would instead be approximated by some estimate $\widehat{w}(x)$, e.g.,
by estimating these unknown densities via fitting models to the training and test data. The error in this estimation process---i.e., the error of
$\widehat{w}$ in estimating $w$---can then
lead to a failure of coverage in the resulting prediction intervals. 

In this paper, we consider a special case of the covariate shift problem: we are interested in the setting where the (unknown) 
weight function $w(x)$ depends only on the subpopulation to which this data point belongs. Specifically, suppose that the feature vector $X$
can be written as $X = (X^0,X^1)$, where $X^0\in[K] = \{1,\dots,K\}$ indicates the group or subpopulation to which the data point belongs,
and $X^1$ captures any remaining features for this data point. 
This type of structure can arise in settings where we collect data from different groups, and can include sampling strategies such
as stratified sampling. We will consider a setting where the weight function $w(x)$ depends only on the first part of the feature, i.e.,
if two data points $X,X'$ belong to the same group (i.e., $X^0= X'{}^0$) then $w(X) = w(X')$. 
As a motivating example, suppose that our data points consist of students within a particular city's school system.
It may be the case that the sampled data (drawn from $P$) are distributed differently across schools $k=1,\dots,K$ 
relative to the general population (which has distribution $Q$), e.g., if the $K$ schools differ in terms of their budget constraints 
or their willingness to participate in the study. However, within each school the study samples students uniformly at random, 
then the difference between the training sample distribution $P$ and the general population (i.e., test data) distribution $Q$
can be captured by a weight function $w(X)$ that depends only on the school to which the individual belongs, i.e., only on the value of $X^0\in [K]$.

The contribution of our paper lies in examining the role of the weight function $w$, and its estimate $\widehat{w}$,
in quantifying the coverage guarantees for prediction intervals constructed via WCP. 
In general, an inaccurate estimate $\widehat{w}$ can lead to a substantial loss of coverage for the WCP prediction intervals.
In the special case of group-weighted conformal prediction,
where $w(X)$ depends only on the group $X^0$ to which the data point belongs, we derive
bounds on the loss of coverage that are far tighter than what is currently known in the existing literature. 

\paragraph{Outline of paper.} 
The remainder of the paper is organized as follows. In Section \ref{section:background}, we give background on the CP and WCP methods,
and the covariate shift setting. In Section \ref{sec:problem_setting}, we introduce the setting of group-wise covariate shift, and explain why the prior guarantees for WCP give weak results in this setting. In Sections \ref{sec:fixed_nk} and \ref{sec:gwcp_cov_shift}, we present our main results
under two different formulations of the problem. Finally, we conclude with a short discussion in Section \ref{sec:discussion}.

\section{Background: conformal prediction under covariate shift}
\label{section:background}
In this section, we give background on the conformal prediction methodology, and its variant, weighted conformal prediction, that allows
for predictive inference in the covariate shift setting.

\subsection{Conformal prediction}
The conformal prediction (CP)framework, which aims to provide prediction intervals with distribution-free validity,
 was initially developed
beginning in the late 1990's (see, e.g., early works by \citep{gammerman1998learning,saunders1999transduction,papadopoulos2002inductive},
and see \citep{vovk2005algorithmic} for a comprehensive overview). The CP methodology 
 quickly gained popularity within the statistics community as a robust tool for regression and classification \citep{lei2018distribution,romano2019conformalized}.
Conformal prediction is a procedure that can be
used in conjunction with any regression procedure or modeling algorithm, to produce prediction intervals with valid marginal coverage guarantees. 

In this section and throughout the paper, we will restrict our attention to one particular version of this procedure, the \emph{split conformal} method (also
called \emph{inductive conformal}) \citep{vovk2005algorithmic,papadopoulos2008inductive}, which separates the process into two steps, model fitting and calibration, each run on a different portion of the data set---this offers a computationally efficient option (relative to the high computational cost of the alternative, \emph{full conformal}, which we do not cover here).
For convenience, we assume that the training data consists of $n^*+n$ many data points, partitioned into data sets
 $\{(X^*_i,Y^*_i)\}_{i\in[n^*]}$ (used for model fitting) and $\{(X_i,Y_i)\}_{i\in[n]}$ (used for calibration). We will refer to these two data sets as the \emph{pretraining set} and the \emph{calibration set}. Of course, in practice we would 
be given a single data set and would then partition it randomly into two parts, i.e., data splitting, but the notation of two separate data sets will be more 
convenient for the exposition of the paper. The split conformal procedure then operates as follows:
\begin{enumerate}
\item \textbf{Model fitting:} using the pretraining data set, $\{(X^*_i,Y^*_i)\}_{i\in[n^*]}$, fit a \emph{score function} $s:\mathcal{X}\times\mathcal{Y}\rightarrow\mathbb{R}$, with large values of the score indicating an unusual data value. In the case of a real-valued response, $\mathcal{Y}=\mathbb{R}$,
a common choice of the score function is the absolute residual $s(x,y) = |y - \widehat{f}(x)|$,
where $\widehat{f}:\mathcal{X}\rightarrow\mathbb{R}$ is a regression function fitted on the pretraining data set.
\item \textbf{Calibration:} using the calibration data set,  $\{(X_i,Y_i)\}_{i\in[n]}$, compute scores $s_i = s(X_i,Y_i)$,
and compute the quantile
\begin{equation}\label{eqn:qhat_splitCP}\widehat{q} = \text{Quantile}_{(1-\alpha)(1+1/n)}\left(s_1,\dots,s_n\right).\end{equation}
This is essentially the $(1-\alpha)$-quantile of the scores on the calibration set, but has a small correction factor that 
accounts for the error incurred by having a finite sample size.
\item \textbf{Prediction:} the prediction interval (or more generally, prediction set) at a new feature value $x\in\mathcal{X}$
is then given by
\[\widehat{C}_n(x) = \left\{y\in\mathcal{Y}: s(x,y)\leq \widehat{q}\right\}.\]
\end{enumerate}
After conditioning on pretraining data set (i.e., so that the score function $s$ can be treated as a fixed function),
the split conformal prediction method offers the following guarantee:
for a test point $(X_{n+1},Y_{n+1})$,
if the data points $(X_1,Y_1),\dots,(X_{n+1},Y_{n+1})$ are exchangeable (e.g., are i.i.d.\ from any distribution),
then
\[\mathbb{P}\left\{Y_{n+1} \in \widehat{C}_n(X_{n+1})\right\} \geq 1 - \alpha.\]
In other words, if the data points  $\{(X_i,Y_i)\}_{i\in[n]}$ used for calibration are drawn from the \emph{same} distribution
as the test data, then predictive coverage is guaranteed.

\subsection{Weighted conformal prediction and covariate shift}\label{sec:background_WCP}
The above guarantee, for split conformal prediction, can be considered to be distribution-free in the sense that the 
calibration and test data points can be drawn from \emph{any} distribution $P$ on $\mathcal{X}\times\mathcal{Y}$.
However, the validity of the method relies on the calibration and test data being drawn from the \emph{same} distribution.
We next turn to the covariate shift setting, which as described above, relaxes this assumption: 
we assume that $(X_1,Y_1),\dots,(X_n,Y_n)$ are drawn i.i.d.\ from some distribution $P$, while the test point $(X_{n+1},Y_{n+1})$
is drawn independently from a potentially different distribution $Q$, but with the restriction that $P_{Y|X} = Q_{Y|X}$.

To address this setting, \citep{tibshirani2019conformal} proposed a modification of CP, the weighted conformal prediction (WCP) method.
Working again in the split conformal framework (so that we can assume a score function $s = s(x,y)$ has been pretrained),
the WCP method modifies the calibration step by computing a \emph{weighted} quantile of scores:
\begin{equation}\label{eqn:qhat_WCP}\widehat{q} = \textnormal{Quantile}_{1-\alpha}\left(\sum_{i=1}^n \frac{w(X_i)}{\sum_{i'=1}^{n+1}w(X_{i'})}\cdot \delta_{s_i} + \frac{w(X_{n+1})}{\sum_{i'=1}^{n+1}w(X_{i'})}\cdot\delta_{+\infty}\right),\end{equation}
where $\delta_s$ is the point mass at $s$, and where $w(x)$ captures the shift from the training feature distribution $P_X$ 
to the test feature distribution $Q_X$. The prediction interval is then again given by $\widehat{C}_n(x) = \left\{y\in\mathcal{Y}: s(x,y)\leq \widehat{q}\right\}$,
as for the unweighted case.
To compare to the unweighted version of the method, we can observe
that the quantile $\widehat{q}$ computed in~\eqref{eqn:qhat_splitCP} is a special case of ~\eqref{eqn:qhat_WCP} with a constant weight function $w(x)\equiv 1$---that is, in the case $P_X=Q_X$, where there is no covariate shift.

In the case where the weight function $w(x)$, which characterizes the covariate shift, is known exactly, the WCP method
offers a guarantee of predictive validity:
\[\mathbb{P}_{P^n\times Q}\left\{Y_{n+1}\in\widehat{C}_n(X_{n+1})\right\}\geq 1-\alpha,\]
where probability is taken with respect to the joint distribution $P^n\times Q$, i.e., the training data points $\{(X_i,Y_i)\}_{i\in[n]}$ used for calibration
are drawn i.i.d.\ from $P$, while the test point $(X_{n+1},Y_{n+1})$ is drawn from $Q$.
Specifically, we are assuming that $w(x) \propto \frac{\mathsf{d}Q_X}{\mathsf{d}P_X}(x)$, the Radon--Nikodym derivative
relating the feature distributions $P_X$ and $Q_X$; for instance, as mentioned before, in the special case that $P_X, Q_X$ have densities $f,g$, respectively 
(with respect to any base measure), this is equivalent to requiring $w(x)\propto g(x)/f(x)$.

The WCP method has subsequently been applied to a range of statistical problems that can be reformulated as instances
of covariate shift,
including applications to survival analysis \citep{gui2022conformalized,candes2023conformalized}, causal inference \citep{lei2021conformal,yin2022conformal,jin2023sensitivity}, and adaptive learning \citep{fannjiang2022conformal}.

\paragraph{WCP with estimated weights.}
In practice, 
a major challenge in the implementation of the WCP procedure is that the weight function 
 $w(x) \propto \frac{\mathsf{d}Q_X}{\mathsf{d}P_X}(x)$ is not known exactly---indeed,
\citep{tibshirani2019conformal}'s initial results for WCP offer no theoretical guarantees for an estimated $w(x)$,
although their empirical results demonstrate that estimating $w$ on a large sample performs well in practice.
More recent results in the literature offer guarantees that quantify the extent to which errors in estimating $w$
can lead to loss of coverage. For example, 
\citep{lei2021conformal} prove the following guarantee. Assume we are given a \emph{fixed} function $\widehat{w}(x)$ (or more generally, a function $\widehat{w}(x)$ that 
is estimated independently of the calibration data and test point and can thus be treated as fixed---for instance, $\widehat{w}(x)$ may be fitted on the training set). If the true and estimated weights are normalized to satisfy $\mathbb{E}_{P_X}[w(X)]= \mathbb{E}_{P_X}[\widehat{w}(X)] = 1$, then
the WCP predictive interval satisfies
\begin{equation}\label{eqn:L1_w_bound}
\mathbb{P}_{P^n\times Q}\left\{Y_{n+1}\in\widehat{C}_n(X_{n+1})\right\}\geq 1-\alpha -\mathbb{E}_{P_X}\left[\frac{\left|\widehat{w}(X) - w(X)\right|}{2}\right].\end{equation}
Related results, also establishing bounds on the loss of coverage in terms of the difference $\widehat{w}-w$,
can be found throughout the literature, e.g., see \citet{candes2023conformalized,gui2022conformalized,yang2022doubly,jin2023sensitivity,yin2022conformal}.

\paragraph{An alternative approach: within-group coverage.}
Finally, we mention another line of the existing literature, which seeks to provide group-conditional coverage type guarantees (i.e., in the notation of this work, guarantees of the type $\mathbb{P}\{Y_{n+1}\in\widehat{C}_n(X_{n+1}) \mid X_{n+1}^0 = k\}$). For instance, works such as \citep{dunn2022distribution, ding2023class} study problems related to this framework. Any guarantee of group-conditional coverage will automatically yield coverage with respect to a group-wise covariate shift---that is, this type of guarantee is strictly stronger than coverage with respect to Q, which we study in this paper. However, a group-conditional guarantee is meaningfully achievable only when within-group sample size, $n_k$, is fairly large, while our work allows for small $n_k$ (as long as $K$, the number of groups, is large). Additional related approaches in the literature include the work of \citep{gupta2020distribution} on bin-wise calibration, \citep{gibbs2023conformal},\citep{jung2022batch} which focus on group conditional coverage guarantees.

\section{Problem setting: group-wise covariate shift}
\label{sec:problem_setting}
As described in Section~\ref{sec:intro}, in this paper we consider the setting where the covariate shift arises from the presence of different groups
or subpopulations within the data, which have different probabilities under the training versus test distributions. Specifically, we assume that
the distribution $P$ of a training data point $(X,Y) = (X^0,X^1,Y)$ has the form
\begin{equation}\label{eqn:define_P_hier}\begin{cases} X^0 \sim \textnormal{Multinomial}(p_1,\dots,p_K),\\
(X^1, Y) \mid (X^0 = k) \sim \Pi_k,\end{cases}\end{equation}
i.e., the training data point is drawn from group $k$ with probability $p_k$, and then conditional on being drawn from group $k$,
the remaining features $X^1$ and response $Y$ are drawn from some joint distribution $\Pi_k$. The test data distribution $Q$
is instead given by
\begin{equation}\label{eqn:define_Q_hier}\begin{cases} X^0 \sim \textnormal{Multinomial}(q_1,\dots,q_K),\\
(X^1, Y) \mid (X^0 = k) \sim \Pi_k,\end{cases}\end{equation}
so that the test point is sampled from group $k$ with probability $q_k$ rather than $p_k$,
but then conditional on being drawn from group $k$, the remaining features and response $(X^1,Y)$ follow the \emph{same}
joint distribution $\Pi_k$. 

In this setting, the weight function $w$ characterizing the covariate shift is simple to compute: at a data point $X = (X^0,X^1)$,
the weight function is given by
\[w(X) \propto q_k/p_k \]
if $X^0 = k$, for each $k\in[K]$. In particular, implementing WCP and thus ensuring a predictive coverage
guarantee relies only on knowing these weights; coverage does not rely on any knowledge
or assumptions for the distributions $\Pi_k$, which determine the feature and response
distributions within each group $k$, since these are assumed to be the same across the training and test data.

For most of this paper, we will assume that the test proportions $q_k$ are known. Effectively, we are assuming that there is ample
unlabeled data i.e., test points, for which we observe $X$ but not $Y$, so that we can estimate the proportion of subgroups
within the target population with extremely high accuracy or we simply want coverage with respect to the empirical assignment of test points to the different groups. In Section~\ref{sec:unknown_qk}, we will extend to the setting where these proportions are unknown, and the $q_k$'s are instead estimated empirically.
Alternatively we can interpret the assumption that the $q_k$'s are known in a different way: we are simply requiring that coverage holds relative to some prespecified distribution over the $K$ groups. For example, if we set $q_k\equiv 1/K$, this indicates that we are interested in the average group-wise coverage.
The problem of estimating a weight function $w(X)$ then reduces to the question of 
estimating the proportions $p_1,\dots,p_K$, to characterize the sampling bias of the training data, drawn from $P$, 
as compared to the target distribution $Q$.

\subsection{Group-weighted conformal prediction (GWCP)}\label{sec:define_GWCP}

We now define the group-weighted conformal prediction (GWCP) method that will be the focus of the remainder of the paper.
As before, we work within a split conformal framework, where we are given a pretrained score function $s=s(x,y)$ (e.g., the residual score, $s(x,y) = |y-\widehat{f}(x)|$, where $\widehat{f}$ is a pretrained regression model), and where we evaluate the scores $s_i = s(X_i,Y_i)$ for the calibration data points, $i\in[n]$.

We define the group-weighted split conformal prediction interval as follows. Let  $n_k = \sum_{i\in[n]}\mathbf{1}_{X^0_i=k}$ denote the number of observations in the calibration set that belong to the $k$th group, and define
\[\widehat{P}_{\textnormal{score}}^{(k)} = \frac{1}{n_k}\sum_{i\in[n], X^0_i = k} \delta_{s_i},\]
which is the empirical distribution of scores for training data points in group $k$, for any $k$ with $n_k>0$. If instead $n_k=0$ (i.e., group $k$ does not appear
at all in the calibration data set), then we simply take $\widehat{P}_{\textnormal{score}}^{(k)} = \delta_{+\infty}$, the point mass at infinity. Finally, define
\begin{equation}\label{eqn:GWCP}
\widehat{C}_n(x) = \left\{y\in\mathcal{Y}: s(x,y) \leq \widehat{q}\right\}\textnormal{ where }\widehat{q} = \textnormal{Quantile}_{1-\alpha}\left(\sum_{k=1}^K q_k \widehat{P}_{\textnormal{score}}^{(k)}\right).
\end{equation}
In other words, the prediction interval is defined via a score threshold given by the $(1-\alpha)$-quantile of the following distribution: compute the empirical distribution $\widehat{P}_{\textnormal{score}}^{(k)}$ of scores from data points in group $k$, and then take the mixture of these empirical distributions by placing weight $q_k$ on the $k$th component, for each group $k\in[K]$. A useful feature of this construction is that we do not require the group assignments of the points in the test set to be known.

\citet{dunn2022distribution} used this kind of CDF pooling technique in the context of a two-layer hierarchical model, but only asymptotic guarantees are shown; this type of approach is also studied by \citet{lee2023distribution}, who show finite sample guarantees, but again in a different setting where new groups are sampled from a hierarchical model.

\paragraph{Comparing  to WCP.}
To better compare to our earlier notation for the WCP method, we can rewrite the threshold $\widehat{q}$ in~\eqref{eqn:GWCP} as follows:
\[
\widehat{q} = \textnormal{Quantile}_{1-\alpha}\left(\sum_{i=1}^n \frac{q_{X^0_i}}{n_{X^0_i}}\cdot \delta_{s_i}\right).
\]
In other words, each data point $i\in[n]$ is given weight $q_{X^0_i}/n_{X^0_i}$ in this weighted quantile calculation.

We will now see that we can view this as a \emph{slightly less conservative} version of the WCP method.
To see why, suppose we let $\widehat{p}_k = n_k/n$---that is, if the calibration data points $(X_1,Y_1),\dots,(X_n,Y_n)$
are sampled from the $K$ groups with probabilities $p_1,\dots,p_K$, as in~\eqref{eqn:define_P_hier}, then $\widehat{p}_k$
is an estimate of $p_k$ for each $k$.
Then, for any $X$ with $X^0=k$, the weight for the covariate shift under this model is given by $w(X) = q_k/p_k$;
we can estimate this weight with $\widehat{w}(X) = q_k / \widehat{p}_k$. 
In this comparison, we will assume for simplicity that $n_k>0$ for all $k$, i.e., all groups are observed in the training set---otherwise
the function $\widehat{w}$ would not be well-defined.
By definition,
 we have $\sum_{i\in[n]}\widehat{w}(X_i) = n$.
Then rewriting the definition of $\widehat{q}$ one more time, 
\begin{equation}\label{eqn:GWCP_alt_notation}
\widehat{q} = \textnormal{Quantile}_{1-\alpha}\left(\sum_{i=1}^n \frac{\widehat{w}(X_i)}{\sum_{i'=1}^{n}\widehat{w}(X_{i'})}\cdot \delta_{s_i}\right).
\end{equation}
In contrast, if we were to run WCP with the weight function $\widehat{w}$,\footnote{While the theory for WCP, described
above in Section~\ref{sec:background_WCP}, requires the estimated weight function $\widehat{w}$ to be prefitted---i.e., independent of the calibration and test data---the 
WCP algorithm itself can nonetheless be defined with a data-dependent $\widehat{w}$.} the prediction set would be given by 
$\widehat{C}_n(x) = \{y\in\mathcal{Y}:s(x,y)\leq \widehat{q}_+\}$, where
\begin{equation}\label{eqn:GWCP_plus}\widehat{q}_+ = \textnormal{Quantile}_{1-\alpha}\left(\sum_{i=1}^n \frac{\widehat{w}(X_i)}{\sum_{i'=1}^{n+1}\widehat{w}(X_{i'})}\cdot \delta_{s_i} + \frac{\widehat{w}(X_{n+1})}{\sum_{i'=1}^{n+1}\widehat{w}(X_{i'})}\cdot\delta_{+\infty}\right).\end{equation}
That is, the definition of the threshold $\widehat{q}$ for GWCP is strictly less conservative
since it is exactly the same as the WCP threshold $\widehat{q}_+$ except that the weight on $\delta_{+\infty}$ has been removed.

As we will
see below, the prediction interval defined here in~\eqref{eqn:GWCP}
leads to a coverage guarantee that is slightly smaller than the target level $1-\alpha$. Interestingly, however, this cannot
be fixed by returning to the slightly more conservative WCP method given in~\eqref{eqn:qhat_WCP}.
We choose to work with this present formulation,~\eqref{eqn:GWCP}, since avoiding the additional weight on $\delta_{+\infty}$
leads to a simpler construction and a (slightly) less conservative method without any additional loss of coverage in
the theoretical guarantee. Of course, any coverage guarantee that can be established for the prediction interval defined in~\eqref{eqn:GWCP}
will also hold for the version of WCP that adds a weight to $\delta_{+\infty}$, since this leads to a strictly more conservative
prediction interval.

\subsection{Estimation error in the group-wise setting: applying prior results}\label{sec:problem_setting_prior_bd}
As mentioned above, in this special case, where the covariate shift is solely determined by the $K$ groups, the problem of estimating $w$ is easier:
in particular, we only need to estimate a $K$-dimensional parameter, i.e., $(p_1,\dots,p_K)$
(since we have assumed the $q_k$'s are known),  in contrast
to the  general
case, where estimating an arbitrary function $w(x)$ is an infinite-dimensional problem.

However, even in this  setting of a finite-dimensional unknown parameter,
 existing coverage guarantees provide a fairly pessimistic view and imply that the loss of coverage
may be severe.
In particular, if our estimate of the probabilities $p_1,\dots,p_K$ (which determine the relative proportions of the $K$ 
groups in  the training distribution $P$) is based on a sample of size $\mathcal{O}(n)$, then
we would expect that each $p_k$ is estimated with error scaling as $\mathcal{O}(\sqrt{p_k/n})$ (due to the variance of a Binomial
distribution); equivalently, the \emph{relative} error in estimating $p_k$ scales as $\mathcal{O}(1/\sqrt{np_k})$.
For example, if $p_k \equiv 1/K$ (i.e., the true distribution places equal weight on each group), then 
the relative error in estimating each $p_k$ scales as $\mathcal{O}(\sqrt{K/n})$, and so
the coverage loss term in \citep{lei2021conformal}'s guarantee~\eqref{eqn:L1_w_bound} can be expected to scale as
\[\mathbb{E}\left[\frac{\left|\widehat{w}(X) - w(X)\right|}{2}\right] \asymp \sqrt{\frac{K}{n}}.\]

To demonstrate this intuition more concretely, we compute \citep{lei2021conformal}'s coverage guarantee~\eqref{eqn:L1_w_bound} empirically,
across a range of sample sizes $n = 100,110,120,\dots,1000$, in three regimes: a constant number of groups, $K=10$; a slowly increasing
number of groups, $K=\lfloor \sqrt{n}\rfloor$; and a proportional number of groups, $K = n/10$. In each case, the probabilities are given
by $p_k = q_k \equiv 1/K$ for all groups $k\in[K]$, for both the training and test distributions.
The estimated weight function $\widehat{w}$ is then given by $\widehat{w}(X) = q_k / \widehat{p}_k$ for any data point
with $X^0 = k$, where $q_k\equiv 1/K$ is known, while $\widehat{p}_k$ is given by the empirical
proportion of group $k$ in a pretraining sample of size $n$. To avoid dividing by zero, we add a +1 adjustment to each group's 
count---that is, $\widehat{p}_k = \frac{n_k+1}{n+K}$,
where $n_k$ is the number of samples observed in group $k$ within the training set of size $n$ (since the true proportions $p_k$ are uniform, this adjustment can only make the estimates more accurate).

Figure~\ref{fig:WCP_demo_L1_bound} illustrates the guaranteed coverage level~\eqref{eqn:L1_w_bound}  for each 
choice of $n$ and $K$, where the term $\mathbb{E}\left[\frac{\left|\widehat{w}(X) - w(X)\right|}{2}\right]$
is estimated empirically over $100$ trials.\footnote{Code for reproducing all figures in the paper is available at \url{https://github.com/aabeshb/Group-Weighted-Conformal-Prediction}.}
We can see that, even though estimating $w$ is a finite-dimensional parameter estimation problem, 
the resulting guarantee suggests that there may be substantial loss of coverage; in the proportional regime, $K=n/10$,
the guaranteed coverage level does not even converge to $1-\alpha$ as $n\rightarrow\infty$.

\begin{figure}[t]
    \centering
    \includegraphics[width = 0.8\textwidth]{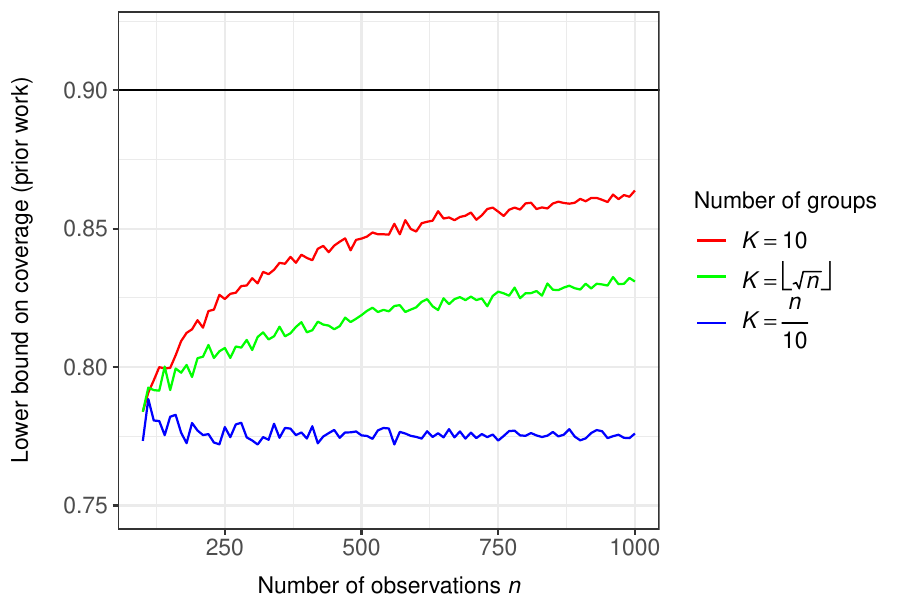}
    \caption{\citep{lei2021conformal}'s WCP coverage guarantee~\eqref{eqn:L1_w_bound} for varying numbers of groups.
    The horizontal line marks the target coverage level 90\% (i.e., $\alpha = 0.1$). See Section~\ref{sec:problem_setting_prior_bd} for details.}
    \label{fig:WCP_demo_L1_bound}
\end{figure}

\section{GWCP with fixed group sizes}\label{sec:fixed_nk}

We now turn to the theoretical properties of the GWCP prediction interval.
We begin by considering a different sampling framework: suppose that the within-group sample sizes $n_1,\dots,n_k \geq 0$ are \emph{fixed} rather than random. 

The problem setting can be formalized as follows: independently for each group $k\in[K]$, we sample $n_k$ many calibration points
from the $k$th group,
\begin{equation}\label{eqn:distrib_fixed_nk}X^0_i=k, \ (X^1_i,Y_i)\stackrel{\textnormal{iid}}{\sim} \Pi_k,\textnormal{ \ for each $i\in\{n_1 + \dots + n_{k-1} + 1, \dots, n_1 + \dots + n_{k-1} + n_k\}$}.\end{equation}
We then want to ensure predictive coverage with respect to a test point $(X,Y) = (X^0,X^1,Y)$ drawn from the target distribution $Q$, defined as in~\eqref{eqn:define_Q_hier}, as before.
At a high level, we can think of this framework as strictly harder than the random setting, because here we will require coverage to hold for any fixed $n_1,\dots,n_K$, while in the random setting described above in Section~\ref{sec:problem_setting}, a marginal coverage guarantee only requires coverage to hold \emph{on average} over the draw of $n_1,\dots,n_K$ (i.e., when the calibration data points are sampled i.i.d.\ from $P$, as in~\eqref{eqn:define_P_hier}).

\subsection{Theoretical guarantee}
Our first main result establishes a marginal coverage guarantee for the fixed group size setting.

\begin{theorem}[Coverage under fixed group sizes.]\label{thm:GWCP_fixed_nk}
Suppose the training data points $(X_1,Y_1),\dots,(X_n,Y_n)$ are sampled from the fixed-group-size model, as in~\eqref{eqn:distrib_fixed_nk},
while the test point $(X_{n+1},Y_{n+1})$ is drawn independently from the distribution $Q$ as defined in~\eqref{eqn:define_Q_hier}.
Then, for any fixed (i.e., pretrained) score function $s:\mathcal{X}\times\mathcal{Y}\rightarrow\mathbb{R}$, the GWCP prediction interval $\widehat{C}_n$ defined in~\eqref{eqn:GWCP} satisfies
    \[\mathbb{P}\left\{Y_{n+1}\in\widehat{C}_n(X_{n+1})\right\} \geq 1 - \alpha - \max_{k:n_k>0} \{q_k/n_k\}.\]
\end{theorem}
In the special case where $q_k \equiv 1/K$, i.e., the target distribution for the test data places equal weight on each of the $K$ groups, 
we can rewrite this guarantee as
\begin{equation}\label{eqn:uniform_case_thm:GWCP_fixed_nk}
\mathbb{P}\left\{Y_{n+1}\in\widehat{C}_n(X_{n+1})\right\} \geq 1 - \alpha - \frac{1}{K  \min_{k:n_k>0} n_k}.\end{equation}
We emphasize that this theorem shows the potential loss of coverage is low as long as \emph{either} the number of groups
$K$ or the minimum group size $n_k$ is fairly large. We will examine the strength of this result more later on, when we compare
to previous guarantees.
Moreover, in Appendix~\ref{app:tight_thm:GWCP_fixed_nk}, we will verify that this lower bound on coverage is tight,
while in Appendix~\ref{app:GWCP_exact_coverage}, we will define a slightly more conservative version of the method 
that can be used if we wish to ensure a coverage
guarantee of at least $1-\alpha$, without an error term.

We now turn to the proof of the theorem.
\begin{proof}[Proof of Theorem~\ref{thm:GWCP_fixed_nk}]   
First we prove the result for the case that $n_k>0$ for all $k$. We begin by defining some notation.
For convenience of the proof, we will redefine the indexing of the training set: 
for each $k$, and for $i=1,\dots, n_k$, let 
\[(X_{k,i},Y_{k,i}) = (X_{n_1 + \dots + n_{k-1} + i},Y_{n_1 + \dots + n_{k-1} + i}).\]
That is, the data points are now indexed within each group, rather than across the entire data set of size $n$.
Next let
\[\widehat{P}^{(k)}_{\textnormal{score}} = \frac{1}{n_k}\sum_{i=1}^{n_k}\delta_{s(X_{k,i},Y_{k,i})}\]
be the empirical distribution of the training data set scores within the $k$th group, as in the definition of the GWCP method, and let
\[\widehat{P}_{\textnormal{score}} = \sum_{k=1}^K q_k \cdot \widehat{P}^{(k)}_{\textnormal{score}}\]
be the weighted mixture across groups. Then the threshold $\widehat{q}$ for GWCP, defined in~\eqref{eqn:GWCP}, 
can equivalently be written as $\widehat{q} = \textnormal{Quantile}_{1-\alpha}(\widehat{P}_{\textnormal{score}})$,
and thus the
 GWCP prediction set is given by
\[\widehat{C}_n(x) = \left\{y\in\mathcal{Y}: s(x,y)\leq \textnormal{Quantile}_{1-\alpha}(\widehat{P}_{\textnormal{score}})\right\}.\]

Next, independently for each $k\in[K]$, sample a new data point in the $k$th group,
\[(X_{k,0},Y_{k,0}) = (k , X^1_{k,0},Y_{k,0}) \textnormal{ where }(X^1_{k,0},Y_{k,0})\sim \Pi_k,\]
independently from the training data.
Then, since the test point $(X_{n+1},Y_{n+1})$ is sampled from the mixture distribution $Q$ which is defined by 
placing weight $q_k$ on each of the $k$ groups, the coverage probability can equivalently be expressed as
\begin{multline*}\mathbb{P}\left\{Y_{n+1}\in\widehat{C}_n(X_{n+1})\right\} 
= \sum_{k=1}^K q_k \mathbb{P}\left\{Y_{k,0}\in\widehat{C}_n(X_{k,0})\right\} \\
= \sum_{k=1}^K q_k \mathbb{P}\left\{s(X_{k,0},Y_{k,0})\leq \textnormal{Quantile}_{1-\alpha}(\widehat{P}_{\textnormal{score}})\right\} .\end{multline*}

Next, for each $k$ and each $i=1,\dots,n_k$, we define a modified version of the weighted empirical distribution $\widehat{P}_{\textnormal{score}}$:
\[\widehat{P}_{\textnormal{score}}^{k,i}=\widehat{P}_{\textnormal{score}} + \frac{q_k}{n_k}\left(\delta_{s(X_{k,0},Y_{k,0})} - \delta_{s(X_{k,i},Y_{k,i})}\right).\]
    In particular, for $i=0$, this is the same distribution as before, i.e., $\widehat{P}_{\textnormal{score}}^{k,0} = \widehat{P}_{\textnormal{score}}$,
    for each $k$. On the other hand, for $i\geq 1$, $\widehat{P}_{\textnormal{score}}^{k,i}$ modifies the original weighted empirical distribution by replacing data point $(X_{k,i},Y_{k,i})$ with data point $(X_{k,0},Y_{k,0})$.
    
With this new notation in place, we can then rewrite the coverage probability as
\begin{equation}\label{eqn:key_step_thm:GWCP_exact_coverage}\mathbb{P}\left\{Y_{n+1}\in\widehat{C}_n(X_{n+1})\right\} = \sum_{k=1}^K q_k\mathbb{P}\left\{s(X_{k,0},Y_{k,0}) \leq \textnormal{Quantile}_{1-\alpha}(\widehat{P}_{\textnormal{score}}^{k,0})\right\}.\end{equation}

For each $k$, since $(X_{k,0},Y_{k,0}),\dots,(X_{k,n_k},Y_{k,n_k})$ are i.i.d.\ and are therefore exchangeable, by symmetry we must have
    \[\mathbb{P}\left\{s(X_{k,0},Y_{k,0}) \leq \textnormal{Quantile}_{1-\alpha}(\widehat{P}_{\textnormal{score}}^{k,0})\right\} = \mathbb{P}\left\{s(X_{k,i},Y_{k,i}) \leq \textnormal{Quantile}_{1-\alpha}(\widehat{P}_{\textnormal{score}}^{k,i})\right\}\]
    for every $i=1,\dots,n_k$, and therefore,
    \[\mathbb{P}\left\{Y_{n+1}\in\widehat{C}_n(X_{n+1})\right\} = \sum_{k=1}^K\sum_{i=1}^{n_k} \frac{q_k}{n_k}\mathbb{P}\left\{s(X_{k,i},Y_{k,i}) \leq \textnormal{Quantile}_{1-\alpha}(\widehat{P}_{\textnormal{score}}^{k,i})\right\}.\]
    Moreover, defining $\alpha' = \alpha + \max_k \{ q_k/n_k\}$, we must have $\textnormal{Quantile}_{1-\alpha}(\widehat{P}_{\textnormal{score}}^{k,i}) \geq \textnormal{Quantile}_{1-\alpha'}(\widehat{P}_{\textnormal{score}})$ for every $k,i$ almost surely.
    This is because the total variation distance between these two weighted empirical distributions is bounded as $\textnormal{d}_{\textnormal{TV}}(\widehat{P}_{\textnormal{score}},\widehat{P}_{\textnormal{score}}^{k,i})\leq  q_k/n_k$ (recall that we have assumed $n_k>0$ for all $k$). Therefore,
    \begin{multline*}\mathbb{P}\left\{Y_{n+1}\in\widehat{C}_n(X_{n+1})\right\} \geq  \sum_{k=1}^K\sum_{i=1}^{n_k} \frac{ q_k}{n_k}\mathbb{P}\left\{s(X_{k,i},Y_{k,i}) \leq \textnormal{Quantile}_{1-\alpha'}(\widehat{P}_{\textnormal{score}})\right\}\\ 
    =\mathbb{E}\left[ \sum_{k=1}^K\sum_{i=1}^{n_k} \frac{ q_k}{n_k}\mathbf{1}\left\{s(X_{k,i},Y_{k,i}) \leq \textnormal{Quantile}_{1-\alpha'}(\widehat{P}_{\textnormal{score}})\right\} \right].\end{multline*}
    Finally, the quantity inside the last expected value must be $\geq 1-\alpha'$, almost surely, by definition of the weighted empirical distribution $\widehat{P}_{\textnormal{score}}$ (see, e.g., \citet[Lemma A.1]{harrison2012conservative}).   

Now we turn to the case where we may have $n_k=0$ for some $k$. Define $q_+ = \sum_{k: n_k>0} q_k$. 
If $1-\alpha > q_+$, then deterministically we must have $\widehat{q}=+\infty$ in the definition of the method~\eqref{eqn:GWCP}, which implies coverage holds with probability $1$. from this point on, then, we will consider the case $1-\alpha \leq q_+$.

Let $Q'$ be the distribution of $(X,Y)\sim Q$ conditional on the event that $X^0=k$ for some $k$ with $n_k>0$, and let $Q''$ be the same conditional on the event that $X^0=k$ for some $k$ with $n_k=0$, such that we can express the target distribution $Q$ as the mixture
    \[Q = q_+ \cdot Q' + (1-q_+)\cdot Q''.\]
    Then
    \[\mathbb{P}_Q\left\{Y_{n+1}\in\widehat{C}_n(X_{n+1})\right\} \geq q_+ \cdot \mathbb{P}_{Q'}\left\{Y_{n+1}\in\widehat{C}_n(X_{n+1})\right\},\]
    where the subscript on $\mathbb{P}$ indicates the distribution from which the test point $(X_{n+1},Y_{n+1})$ is drawn.
    Next observe that by construction, we have
    \[\widehat{P}_{\textnormal{score}} = q_+ \cdot \widehat{P}_{\textnormal{score}}' + (1-q_+) \cdot \delta_{+\infty}\]
    where $\widehat{P}_{\textnormal{score}}'$ is the average of empirical distributions for all \emph{observed} groups,
    \[\widehat{P}_{\textnormal{score}}' = \sum_{k:n_k>0} q'_k\cdot \widehat{P}^{(k)} \textnormal{ where }q'_k = q_k/q_+.\]
    We therefore have
    \[\textnormal{Quantile}_{1-\alpha}(\widehat{P}_{\textnormal{score}}) = \textnormal{Quantile}_{1-\alpha'}(\widehat{P}_{\textnormal{score}}')\]
    where 
    $1 - \alpha' = (1-\alpha)/q_+$.
    Therefore,
    \[\widehat{C}_n(x) = \left\{y\in\mathcal{Y}:s(x,y)\leq\textnormal{Quantile}_{1-\alpha'}(\widehat{P}_{\textnormal{score}}')\right\}.\]
    We then have
    \[\mathbb{P}_{Q'}\left\{Y_{n+1}\in\widehat{C}_n(X_{n+1})\right\} \geq 1-\alpha' -\max_{k:n_k>0}\{q_k/n_k\},\]
    which holds by applying our proof above (for the case where $n_k>0$ for all $k$) to the target distribution $Q'$ in place of $Q$. Finally, returning to the original distribution $q$, we therefore have
    \[\mathbb{P}_Q\left\{Y_{n+1}\in\widehat{C}_n(X_{n+1})\right\} \geq q_+ \cdot \left[1-\alpha' - \min_{k:n_k>0}\{q'_k/n_k\}\right] = 1 - \alpha - \max_{k:n_k>0}\{q_k/n_k\},\]
    which completes the proof.
\end{proof}

\subsection{Simulation with fixed group sizes}\label{sec:sim_fixed_nk}
The theoretical guarantee given in Theorem~\ref{thm:GWCP_fixed_nk} above for the fixed group size setting
suggests that the potential coverage loss of the GWCP method depends on the \emph{smallest} nonzero group size---specifically,
in the case where the test distribution is uniform over the $K$ groups, the loss of coverage
in the guarantee~\eqref{eqn:uniform_case_thm:GWCP_fixed_nk}
depends on $\min_{k:n_k>0} n_k$.
It is intuitive that if many groups are small, then we should expect to see a loss of coverage due to issues of low effective sample size,
but we might suspect that the guarantee for GWCP should depend on some kind of \emph{average} measure of group size, rather than 
the minimum. In fact, however, the lower bound in Theorem~\ref{thm:GWCP_fixed_nk} is essentially tight: in a worst case scenario,
the loss of coverage can indeed be governed by the smallest group size, even if most groups are large.

\begin{figure}[t]
    \centering
    \includegraphics[width = \textwidth]{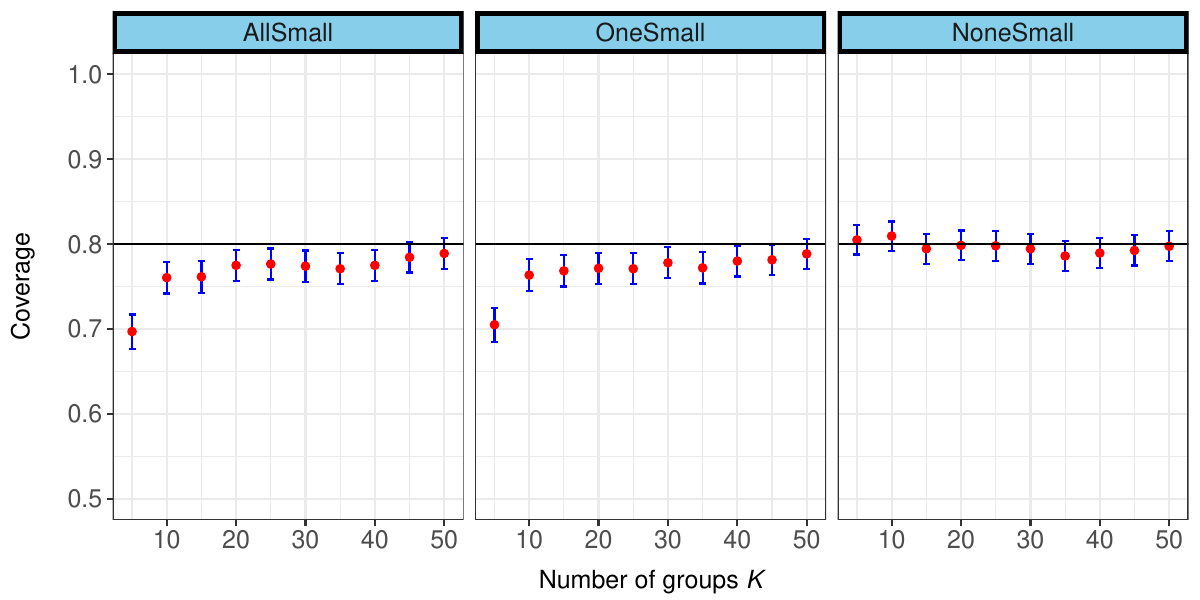}
    \caption{Simulation results for the setting of fixed group sizes, in three regimes for the group sizes $n_1,\dots,n_K$. 
    The horizontal line marks the target coverage level $1-\alpha = 80\%$. Results are averaged over $2000$ independent trials, and the figure displays the mean coverage rate along with a 95\% confidence interval. See Section~\ref{sec:sim_fixed_nk} for details.}
    \label{fig:sim_fixed_nk}
\end{figure}

To illustrate this, we simulate data with fixed group sizes, in three regimes:
\texttt{AllSmall}, where all the group sizes $n_k$ are small; \texttt{OneSmall}, where one of the $n_k$'s is small but the remaining $K-1$ 
values are large; and \texttt{NoneSmall}, where all $n_k$'s are large.
If we believe that coverage loss is controlled by average group size, then we would expect to see the results of the \texttt{OneSmall}
setting to resemble the \texttt{NoneSmall} setting. However, the simulation shows that the loss of coverage in the \texttt{OneSmall} setting
behaves very similarly to the \texttt{AllSmall} setting, confirming that the guarantee of Theorem~\ref{thm:GWCP_fixed_nk}
is tight. (In Appendix~\ref{app:tight_thm:GWCP_fixed_nk}, we will discuss the question of tightness of the lower bound in more theoretical
detail.)

The simulation is designed as follows. The target coverage level is $1-\alpha = 0.8$.
The simulation is repeated for each value $K\in\{5,10,15,\dots,50\}$ as the number of groups.
The distribution of $(X,Y)$ in the $k$th group is given by $\delta_k \times \textnormal{Uniform}[\frac{k-1}{K},\frac{k}{K}]$ (i.e., $X = X^0$ is simply
equal to the group label, $k$, and $Y$ is sampled from the $\textnormal{Uniform}[\frac{k-1}{K},\frac{k}{K}]$ distribution).
The pretrained score function is given by $s(x,y) = y$.
The group sizes are defined for each of the three regimes as:
\begin{itemize}
\item \texttt{AllSmall} regime: $n_k = 1$ for all $k$.
\item \texttt{OneSmall} regime: $n_{0.8K + 1} = 1$, and $n_k = 100$ for all $k\neq 0.8K+1$.
\item \texttt{NoneSmall} regime: $n_k = 100$ for all $k$.
\end{itemize}

The empirical coverage of the resulting prediction set is then computed based on $2000$ many independent trials where for each trial a calibration data set is drawn from the training data distribution $P$ and a test point is drawn from $Q$. Figure~\ref{fig:sim_fixed_nk} shows the mean coverage along with the $95\%$ confidence interval for each setting.

We observe that the empirical coverage levels for the  \texttt{OneSmall} regime behave very similarly to those
for the \texttt{AllSmall} regime, with substantial undercoverage when $K$ is small. 
That is, for the \texttt{OneSmall} regime, the minimum group size, $\min_k n_k =1$,
indeed seems to determine its behavior in terms of a loss of coverage in the finite sample regime. 
Of course, the theoretical guarantee ensures that coverage will converge to (at least) $1-\alpha$ as the number of groups $K$ increases,
even if $\min_k n_k=1$; our simulation results confirm this, since all three regimes show coverage at approximately the target level $1-\alpha$
once $K$ is large.

However, while this result confirms that our theoretical guarantee is tight in the worst case, the specific construction of the \texttt{OneSmall}
regime is designed specifically to be as challenging as possible. The specific design of the simulation
ensures that the smallest group (i.e., the $k$ for which $n_k=1$) is likely to have its score distribution concentrated near the
$(1-\alpha)$-quantile of the overall distribution of scores. This type of worst-case behavior may be unlikely to occur in practice;
while the $\min_k n_k$ term in the theoretical guarantee cannot be improved without further assumptions (see Appendix~\ref{app:tight_thm:GWCP_fixed_nk} for more details),
it is likely that for a typical problem, loss of coverage may depend more on average group size than on minimum group size.

\section{GWCP under covariate shift} \label{sec:gwcp_cov_shift}
We now return to the original problem that motivates this paper: the setting of a group-wise covariate shift.
Returning to our earlier problem formulation, we assume that the $n$ training points are sampled i.i.d.\ from $P$, and the target
test distribution is $Q$, where these distributions are defined as in~\eqref{eqn:define_P_hier} and~\eqref{eqn:define_Q_hier} above.
We will next present our main results for finite-sample coverage guarantees in the case of a group-wise covariate shift,
and then compare to existing coverage guarantees in the prior literature.
\subsection{Theoretical guarantee}
We now give the theoretical guarantee for the group-wise covariate shift setting.
Our theoretical guarantee for the group-wise covariate shift setting is
 a straightforward consequence of the guarantee for the fixed group size
setting: since Theorem~\ref{thm:GWCP_fixed_nk} ensures that coverage holds for \emph{any} group sizes
$n_k$, this means that coverage also holds \emph{on average} over a random draw of the group sizes $n_k$.

\begin{theorem}[Coverage under group-wise covariate shift.]\label{thm:GWCP_cov_shift}
Suppose the training data points $(X_1,Y_1),\dots,(X_n,Y_n)$ are sampled i.i.d.\ from the distribution $P$ as defined in~\eqref{eqn:define_P_hier},
while the test point $(X_{n+1},Y_{n+1})$ is drawn independently from the distribution $Q$ as defined in~\eqref{eqn:define_Q_hier}.
Then, for any fixed (i.e., pretrained) score function $s:\mathcal{X}\times\mathcal{Y}\rightarrow\mathbb{R}$, the GWCP prediction interval $\widehat{C}_n$ defined in~\eqref{eqn:GWCP} satisfies
    \[\mathbb{P}\left\{Y_{n+1}\in\widehat{C}_n(X_{n+1})\right\} \geq 1 - \alpha - \mathbb{E}\left[\max_{k:n_k>0} \{q_k/n_k\}\right].\]
 In particular, if 
 $\min_k p_k \geq \frac{8\log n}{n}$
 then
    \[\mathbb{P}\left\{Y_{n+1}\in\widehat{C}_n(X_{n+1})\right\} \geq 1 - \alpha - 4n^{-1}\max_k \{q_k/p_k\}.\]
   \end{theorem}

 The proof is given in Appendix~\ref{app:proofs_GWCP_cov_shift}; essentially it is achieved by simply applying the guarantee
 for fixed $n_k$'s, given in Theorem~\ref{thm:GWCP_fixed_nk}, and then marginalizing over the random draw of the $n_k$'s.
The lower bound on $p_k$ of $O\left(\frac{\log n}{n}\right)$ is a mild condition that ensures that all groups are observed with high probability.

\subsection{Uniform groups and unobserved groups: a closer look}
In the results above, we have assumed that the group probabilities $q_1,\dots,q_K$ under the target distribution (i.e., the distribution
of the test data) are known. For example, if we know ahead of time that we will need to predict the response $Y$ for a particular
collection of $X$ values, $X_{n+1},\dots,X_{n+m}$, we can simply take the $q_k$'s to be the \emph{empirical} 
proportions of groups $1,\dots,K$ within this list, to guarantee coverage on average over this test set---that is, to
guarantee
\[\frac{1}{m}\sum_{i=1}^m \mathbb{P}\{Y_{n+i}\in\widehat{C}_n(X_{n+i})\}\geq 1-\alpha.\]

On the other hand, we can also imagine settings in which the set of test feature vectors is not available in advance.
In this case, perhaps we do not even know ahead of time the total number of groups in the population.
Suppose, say, we want to ensure coverage relative to a uniform distribution over all groups (i.e., $q_k \equiv 1/K$),
but we do not know the total number of groups $K$; instead, we simply run GWCP with the target distribution given 
by a uniform distribution over all \emph{observed} groups.
This leads to the following modification of the GWCP method~\eqref{eqn:GWCP}:
\begin{equation}\label{eqn:GWCP_unif_unobs}
\widehat{C}_n(x) = \left\{y\in\mathcal{Y}: s(x,y) \leq \widehat{q}\right\}\textnormal{ where }\widehat{q} = \textnormal{Quantile}_{1-\alpha}\left(\sum_{k: n_k \geq 1}  \frac{1}{\widehat{K}}\widehat{P}_{\textnormal{score}}^{(k)}\right) ,\ \widehat{K} = \sum_{k=1}^K\mathbf{1}_{n_k>0}.
\end{equation}

\begin{corollary}\label{cor:GWCP_cov_shift_unobs_grp}
Suppose the training data points $(X_1,Y_1),\dots,(X_n,Y_n)$ are sampled i.i.d.\ from the distribution $P$ as defined in~\eqref{eqn:define_P_hier},
while the test point $(X_{n+1},Y_{n+1})$ is drawn independently from the distribution $Q$ as defined in~\eqref{eqn:define_Q_hier} with $q_k \equiv 1/K$.
Then, for any fixed (i.e., pretrained) score function $s:\mathcal{X}\times\mathcal{Y}\rightarrow\mathbb{R}$, the prediction interval $\widehat{C}_n$ defined in~\eqref{eqn:GWCP_unif_unobs} satisfies
    \[\mathbb{P}\left\{Y_{n+1}\in\widehat{C}_n(X_{n+1})\right\} \geq (1-\alpha) - \mathbb{E}\left[ \frac{1}{K\min_{k:n_k>0} n_k} + (1-\alpha) \frac{\sum_k \mathbf{1}_{n_k=0}}{K}\right].\]
 In particular, if 
 $\min_k p_k \geq \frac{8\log n}{n}$
 then
    \[\mathbb{P}\left\{Y_{n+1}\in\widehat{C}_n(X_{n+1})\right\} \geq 1 - \alpha - \frac{4}{Kn\min_k p_k}.\]
\end{corollary}
The proof is given in Appendix~\ref{app:proofs_GWCP_cov_shift}.

\subsection{Unknown test group probabilities}
\label{sec:unknown_qk}
In this section, we consider the setting where the the group probabilities in the test distribution (i.e., $q_1,q_2,\ldots,q_K$) are not known. Instead we have $m$ many draws $\tilde{X}^0_1,\ldots,\tilde{X}^0_m\in \{1,\dots,K\}$ that offer an empirical approximation to this distribution over groups---for instance, we might have a sample of size $m$ from the test population, and the $\tilde{X}^0_i$'s are the group assignments of these $m$ data points. 

Can we simply use the empirical proportions in each group $k=1,\dots,K$ in place of the $q_k$'s when constructing the GWCP prediction set~\eqref{eqn:GWCP}? That is, our empirically weighted prediction set is given by
\begin{equation}\label{eqn:GWCP_unknown_q}
\widehat{C}_n(x) = \left\{y\in\mathcal{Y}: s(x,y) \leq \widehat{q}\right\}\textnormal{ where }\widehat{q} = \textnormal{Quantile}_{1-\alpha}\left(\sum_{k=1}^K \tilde{q}_k \widehat{P}_{\textnormal{score}}^{(k)}\right),\ \tilde{q}_k = \frac{\sum_{i=1}^m \mathbf{1}\{\tilde{X}^0_i = k\}}{m}.
\end{equation} 
The following result verifies that this approach offers nearly the same level of coverage as before---the additional loss of coverage in the guarantee is only $\frac{1}{m+1}$.
\begin{corollary}\label{cor:GWCP_cov_shift_unknown_q}
Suppose the training data points $(X_1,Y_1),\dots,(X_n,Y_n)$ are sampled i.i.d.\ from the distribution $P$ as defined in~\eqref{eqn:define_P_hier}, and the test point is independently drawn as $(X_{n+1},Y_{n+1})\sim Q$ as defined in~\eqref{eqn:define_Q_hier}. Moreover, suppose $\tilde{X}^0_1,\dots,\tilde{X}^0_m$ are i.i.d.\ draws from the Multinomial$(q_1,\dots,q_K)$ distribution, drawn independently of the data. Then, for any fixed (i.e., pretrained) score function $s:\mathcal{X}\times\mathcal{Y}\rightarrow\mathbb{R}$, the prediction interval $\widehat{C}_n$ defined in~\eqref{eqn:GWCP_unknown_q} satisfies
    \[\mathbb{P}\left\{Y_{n+m+1}\in\widehat{C}_n(X_{n+m+1})\right\} \geq 1-\alpha - \frac{1}{m+1} - \mathbb{E}\left[\max_{k:n_k>0} \{\tilde{q}_k/n_k\}\right].\]
\end{corollary}
The proof is given in Appendix~\ref{app:proofs_GWCP_cov_shift}.

\subsection{Comparison to existing guarantees}\label{sec:compare_lower_bds}
We next provide an empirical comparison between the theoretical guarantee obtained in Corollary~\ref{cor:GWCP_cov_shift_unobs_grp},
 which provides a lower bound for predictive coverage
under group-wise covariate shift, against the existing guarantee obtained by \citep{lei2021conformal}, shown above in~\eqref{eqn:L1_w_bound},
where the loss of coverage is bounded as $\mathbb{E}\left[\frac{\left|\widehat{w}(X) - w(X)\right|}{2}\right]$.
We use the same settings as in Section~\ref{sec:problem_setting_prior_bd}, with three regimes: $K=10$ (a constant number of groups),
 $K= \lfloor \sqrt{n}\rfloor$ (a slowly growing number of groups), and $K=n/10$ (a proportional number of groups). The training and test distributions
 are again uniform, $p_k = q_k \equiv 1/K$.
 Here the coverage guarantee given by \citep{lei2021conformal} is the same as that shown in Figure~\ref{fig:WCP_demo_L1_bound}, above;
our new result, the coverage guarantee of Corollary~\ref{cor:GWCP_cov_shift_unobs_grp}, is now added to the plot for comparison.
Note that this guarantee is for the version~\eqref{eqn:GWCP_unif_unobs} of the group-weighted conformal prediction interval, which does not assume
that we know the total number of groups.

\begin{figure}[t]
    \centering
    \includegraphics[width = \textwidth]{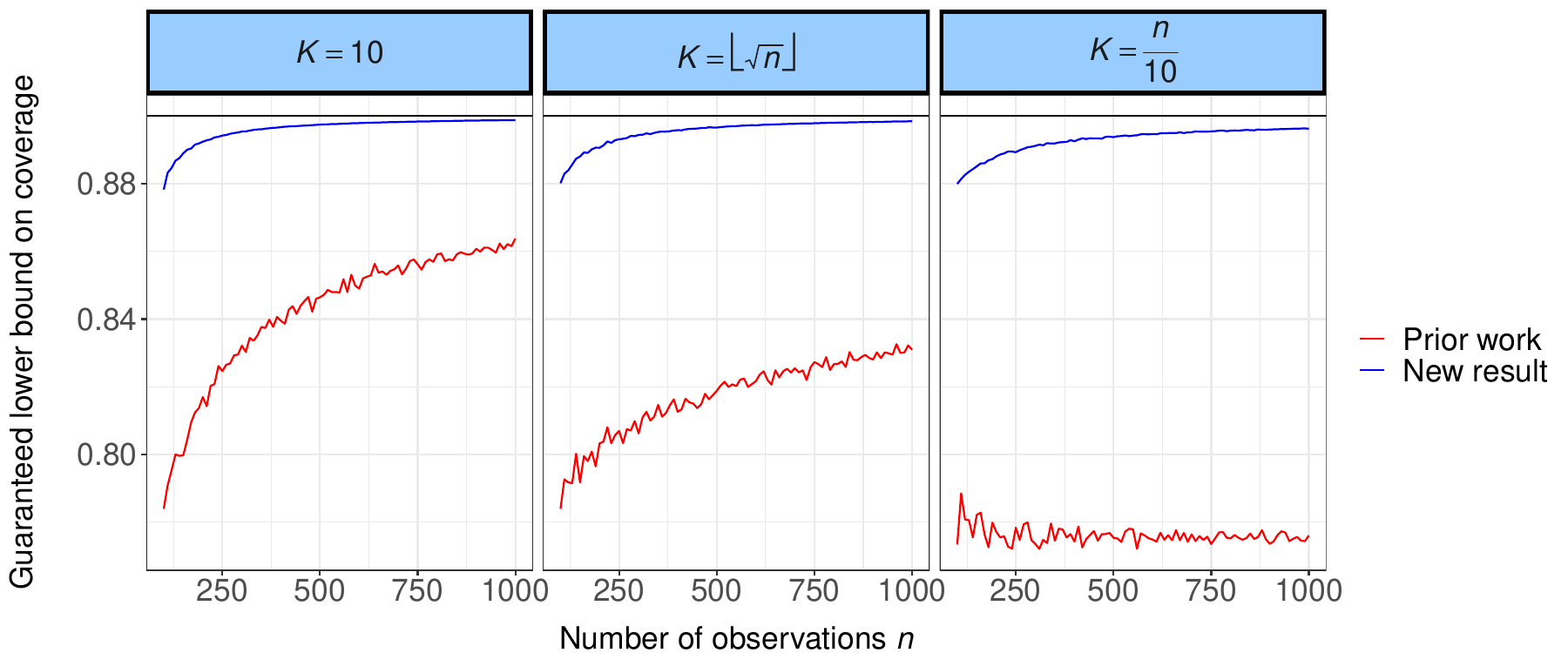}
        \caption{\citep{lei2021conformal}'s WCP coverage guarantee~\eqref{eqn:L1_w_bound}, compared to the new guarantee
        given in Corollary~\ref{cor:GWCP_cov_shift_unobs_grp}, for varying numbers of groups.
    The horizontal line marks the target coverage level 90\% (i.e., $\alpha = 0.1$). See Section~\ref{sec:compare_lower_bds} for details.}
    \label{fig:lower_bound_compare}
\end{figure}

As in Section~\ref{sec:problem_setting_prior_bd}, to ensure that $\widehat{w}$ is always well-defined even though there is a nonzero
probability that not all groups are observed, we define $\widehat{w}(x)$ using estimated proportions $\widehat{p}_k= \frac{n_k+1}{n+K}$,
i.e., augmenting each group count by 1. Therefore, we are giving two advantages to  \citep{lei2021conformal}'s existing bound
relative to ours: (1) the definition of $\widehat{p}_k$, which (since the true $p_k$'s are uniform) is more accurate than the actual empirical
distribution; (2) \citep{lei2021conformal}'s coverage guarantee applies to the more conservative form of the method, given in~\eqref{eqn:GWCP_plus},
relative to ours.

We can see in Figure~\ref{fig:lower_bound_compare} that the new coverage guarantee in Corollary~\ref{cor:GWCP_cov_shift_unobs_grp}
is substantially stronger, converging much more rapidly to the target coverage level of 90\% relative to the existing lower bound. Indeed,
as observed earlier, in the third setting, $K=n/10$, where the sample size within a group does not increase with $n$,
the existing lower bound does not approach the target coverage level $1-\alpha$ even as $n\rightarrow\infty$; in contrast,
we can observe that the lower bound in Corollary~\ref{cor:GWCP_cov_shift_unobs_grp} converges to $1-\alpha$ as long as \emph{either}
the number of groups $K$ or the (minimum) sample size within a group $n_k$ diverges to infinity.

Finally, we remark that another difference between our work and the existing literature is that, for our theoretical guarantee, it is important that the weight function $\widehat{w}$ is estimated using the calibration set---in contrast, in existing work, it is more common to assume $\widehat{w}$ is pretrained, e.g., fitted to the pretraining data. We discuss this distinction more in Appendix~\ref{app:wcp_which_data set}.

\section{Discussion}\label{sec:discussion}
This paper considers a supervised learning setting where the underlying population consists of various groups, with the relative proportions of the different groups being potentially different between the training data (i.e., the observed sample) and the test data (i.e., the target distribution or general population).
We view this as a covariate shift problem, where the distribution shift depends only on a finitely-valued covariate, i.e., the group to which an observation belongs. Thus, the prediction problem can be addressed with weighted conformal prediction, which offers distribution-free coverage guarantees as long as the covariate shift can be accurately estimated. Our results provide sharp coverage guarantees that bound the influence of this estimation error, 
and are substantially closer to the target coverage level $1-\alpha$ than previously known bounds; in particular, our results can even achieve
coverage converging to $1-\alpha$ when the group sizes remain constant as sample size $n$ increases (i.e., the number of groups is proportional to $n$).

More broadly, this problem can be viewed as an instance of a more general class of settings, where covariate shift can  be expressed
via a finite dimensional parametrization. Whether these techniques can be extended to broader scenarios, moving beyond
the setting of disjoint subpopulations, is an important direction for further exploration. Another future direction of interest is the problem of allowing for noisy or corrupted group labels in the test or training data, to handle scenarios where group labels may not be known exactly---for instance, due to issues such as privacy constraints or inaccurate data collection.

\section*{Acknowledgments}
R.F.B. was partially supported by the National Science Foundation via grant DMS-2023109, and by the Office of Naval Research via grant N00014-20-1-2337.

\bibliographystyle{plainnat}  
\bibliography{ref}

\appendix

\section{Additional results and calculations}

\subsection{Additional proofs}\label{app:proofs_GWCP_cov_shift}

\begin{proof}[Proof of Theorem~\ref{thm:GWCP_cov_shift}]
When we sample training data i.i.d.\ from the distribution $P$, as defined in~\eqref{eqn:define_P_hier}, 
this is equivalent to first sampling group sizes $n_1,\dots,n_K$ via the $\textnormal{Multinomial}(p_1,\dots,p_K)$ distribution,
then drawing the appropriate number of samples from each $\Pi_k$ to define our training set.
In other words, conditional on $n_1,\dots,n_K$, 
the distribution of the training and test data is identical to that assumed in Theorem~\ref{thm:GWCP_fixed_nk}. 
Therefore, applying Theorem~\ref{thm:GWCP_fixed_nk}, we have
    \[\mathbb{P}\left\{Y_{n+1}\in\widehat{C}_n(X_{n+1}) \ \big\vert \ n_1,\dots,n_K \right\} \geq 1 - \alpha - \max_{k:n_k>0} \{q_k/n_k\}.\]
 Taking an expected value to marginalize over $n_1,\dots,n_K$ completes the proof of the first part.

To complete the proof we need to verify that $\mathbb{E}\left[\max_{k:n_k>0} \{q_k/n_k\}\right]\leq 4n^{-1}\max_k \{q_k/p_k\}$, under the assumption $\min_k p_k \geq \frac{8\log n}{n}$.
Let $\mathcal{E}$ be the event that $n_k \geq np_k (1-1/\sqrt{2})$ holds for all $k\in[K]$.
On the event $\mathcal{E}$, we have
\[\max_{k:n_k>0} \{q_k/n_k\} \leq \max_k\left\{\frac{q_k}{np_k (1-1/\sqrt{2})} \right\}= (1-1/\sqrt{2})^{-1} n^{-1} \max_k \{q_k/p_k\},\]
while on $\mathcal{E}^c$ we have
$\max_{k:n_k>0} \{q_k/n_k\} \leq \max_k q_k \leq 1 \leq \max_k \{q_k/p_k\}$. Therefore,
\[\mathbb{E}\left[\max_{k: n_k>0} \{q_k/n_k\}\right] 
\leq \max_k \{q_k/p_k\} \left[ (1-1/\sqrt{2})^{-1}n^{-1}+ \mathbb{P}\{\mathcal{E}^c\}\right].\]
Next, applying the multiplicative Chernoff bound to the random variable $n_k\sim\textnormal{Binomial}(n,p_k)$, for each $k$ we have
\[\mathbb{P}\left\{ n_k \leq  np_k(1-1/\sqrt{2})\right\} \leq  e^{-(1-(1-1/\sqrt{2}))^2 \cdot np_k / 2} = e^{-np_k/4} \leq e^{-8\log n/4} =n^{-2},\]
since $p_k \geq \frac{8\log n}{n}$ by assumption. Thus by the union bound,
$\mathbb{P}\{ \mathcal{E}^c\} \leq  Kn^{-2} \leq 1/(8n)$,
where for the last step, we use the bound $K\leq (\min_k p_k)^{-1} \leq \frac{n}{8\log n} \leq n/8$  (we can assume $n\geq 3$ since otherwise
the last bound in the theorem holds trivially). Combining these calculations proves the bound.
\end{proof}

\begin{proof}[Proof of Corollary~\ref{cor:GWCP_cov_shift_unobs_grp}]
First, condition on $n_1,\dots,n_K$, and let $Q'$ be the distribution placing uniform weights on all observed groups, i.e., 

\begin{equation}\label{eqn:define_Q_hier_obs}\begin{cases} X^0 \sim \textnormal{Uniform}(k: n_k \geq 1),\\
(X^1, Y) \mid (X^0 = k) \sim \Pi_k,\end{cases}\end{equation}
As in the proof of Theorem~\ref{thm:GWCP_cov_shift}, applying Theorem~\ref{thm:GWCP_fixed_nk} to a test point drawn from $Q'$, we have
    \[\mathbb{P}_{Q'}\left\{Y_{n+1}\in\widehat{C}_n(X_{n+1}) \ \big\vert \ n_1,\dots,n_K \right\} \geq 1 - \alpha - \max_{k:n_k>0} \{q_k/n_k\}
    =1 - \alpha - \frac{1}{\widehat{K}\min_{k:n_k>0}n_k}.\]
Next, continuing to condition on $n_1,\dots,n_K$, observe that $Q$ is equal to a mixture placing weight $\frac{\widehat{K}}{K}$ on $Q'$
and $1-\frac{\widehat{K}}{K}$ on the remaining groups.
Thus, for a test point $(X_{n+1},Y_{n+1})$ drawn from $Q$, we now have
\begin{multline*}\mathbb{P}\left\{Y_{n+1}\in\widehat{C}_n(X_{n+1}) \ \big\vert \ n_1,\dots,n_K \right\} \geq \frac{\widehat{K}}{K} \cdot \mathbb{P}_{Q'}\left\{Y_{n+1}\in\widehat{C}_n(X_{n+1}) \ \big\vert \ n_1,\dots,n_K \right\}\\ \geq 
    \frac{\widehat{K}}{K}\left(1 - \alpha - \frac{1}{\widehat{K}\min_{k:n_k>0}n_k}\right)
    = (1-\alpha) \cdot \frac{\widehat{K}}{K} - \frac{1}{K\min_{k:n_k>0} n_k}.\end{multline*}
 Taking an expected value to marginalize over $n_1,\dots,n_K$, then,
\begin{multline*}\mathbb{P}\left\{Y_{n+1}\in\widehat{C}_n(X_{n+1})\right\}
\geq \mathbb{E}\left[(1-\alpha) \cdot \frac{\widehat{K}}{K} - \frac{1}{K\min_{k:n_k>0} n_k}\right]\\ = (1-\alpha) - \mathbb{E}\left[ \frac{1}{K\min_{k:n_k>0} n_k} + (1-\alpha) \frac{\sum_k \mathbf{1}_{n_k=0}}{K}\right] .\end{multline*}

Next, we add the assumption that $\min_k p_k \geq \frac{8\log n}{n}$. 
As in the proof of Theorem~\ref{thm:GWCP_cov_shift}, let $\mathcal{E}$ be the event that
$n_k \geq np_k (1-1/\sqrt{2})$ holds for all $k\in[K]$, with $\mathbb{P}\{ \mathcal{E}^c\} \leq 1/(8n)$ as before.
On the event $\mathcal{E}^c$ we have
\[ \frac{1}{K\min_{k:n_k>0} n_k} + (1-\alpha) \frac{\sum_k \mathbf{1}_{n_k=0}}{K} \leq \frac{1}{K} + \frac{K-1}{K} = 1,\]
while on the event $\mathcal{E}$, we have $\sum_k \mathbf{1}_{n_k=0} = 0$, and $\frac{1}{K\min_{k:n_k>0} n_k}\leq \frac{1}{K\min_k np_k (1-1/\sqrt{2})} = \frac{1}{Kn \min_k p_k}\cdot\frac{1}{1-1/\sqrt{2}} $.
Therefore,
\[\mathbb{E}\left[ \frac{1}{K\min_{k:n_k>0} n_k} + (1-\alpha) \frac{\sum_k \mathbf{1}_{n_k=0}}{K}\right] \leq \frac{1}{Kn \min_k p_k}\cdot\frac{1}{1-1/\sqrt{2}} + \frac{1}{8n} \leq \frac{4}{Kn\min_k p_k},\]
where the last step holds since $K\min_k p_k \leq 1$. This completes the proof.
\end{proof}

\begin{proof}[Proof of Corollary~\ref{cor:GWCP_cov_shift_unknown_q}]
First we consider an equivalent representation of the problem. Suppose we have $m+1$ test points $(X_{n+1},Y_{n+1}),(X_{n+2},Y_{n+2}),\dots,(X_{n+m+1},Y_{n+m+1})$ sampled i.i.d.\ from $Q$; we will examine the coverage guarantee for the first test point $(X_{n+1},Y_{n+1})$ as before, while taking the remaining $m$ test points as our source of information for estimating the $q_k$'s---that is, we will define
\[\tilde{q}_k = \frac{\sum_{i=1}^m\mathbf{1}\{X^0_{n+1+i} = k\}}{m}.\]

Now let $\tilde{Q} = \sum_{k=1}^K \tilde{q}_k\cdot \widehat{P}_{\text{score}}^{(k)}$, and let $\tilde{Q}^* = \sum_{k=1}^K \tilde{q}^*_k\cdot \widehat{P}_{\text{score}}^{(k)}$ where 
\[\tilde{q}^*_k = \frac{\sum_{i=1}^{m+1} \mathbf{1}\{X^0_{n+i} = k\}}{m + 1}\]
is the empirical frequency of group $k$ in the entire test set (i.e., the test point of interest, $(X_{n+1},Y_{n+1})$, along with the $m$ test points used to construct the $\tilde{q}_k$'s). 
Note that by definition we have
\[\widehat{q} = \textnormal{Quantile}_{1-\alpha}(\tilde{Q}).\]
Furthermore,
\[\tilde{q}^*_k = \frac{m}{m+1} \tilde{q}_k + \frac{\mathbf{1}\{X^0_{n+1} = k\}}{m+1},\]
and in particular, this means that we can calculate
\[
    \textnormal{d}_{\text{TV}}\left(\tilde{Q},\tilde{Q}^*\right) = 
\frac{1}{2}\sum_{k=1}^K |\tilde{q}_k - \tilde{q}^*_k|
= 
\frac{1}{2}\sum_{k=1}^K \frac{1}{m+1}\left|\tilde{q}_k - \mathbf{1}\{X^0_{n+1}=k\}\right| \leq \frac{1}{m+1}.\]

Thus, $\widehat{q} = \textnormal{Quantile}_{1-\alpha}(\tilde{Q})  \geq \textnormal{Quantile}_{1-\alpha^\prime}(\tilde{Q}^*)$ where $\alpha^\prime = \alpha + \frac{1}{m+1}$, which further implies that

\begin{equation*}
    \widehat{C}_n^*(x) := \left\{y\in\mathcal{Y}: s(x,y) \leq \textnormal{Quantile}_{1-\alpha^\prime}(\tilde{Q}^*)\right\} \subseteq \widehat{C}_n(x)
\end{equation*}.
Therefore, 
$$\mathbb{P}\left\{Y_{n+1}\in\widehat{C}_n(X_{n+1})\right\} \geq \mathbb{P}\left\{Y_{n+1}\in\widehat{C}_n^*(X_{n+1})\right\} .$$
But by symmetry, since the $\tilde{q}_k$'s are constructed as a symmetric function of the $m+1$ test points, we have
\[\mathbb{P}\left\{Y_{n+1}\in\widehat{C}_n^*(X_{n+1})\right\} = \mathbb{P}\left\{Y_{n+i}\in\widehat{C}_n^*(X_{n+i})\right\}\]
for all $i=1,\dots,m+1$. Taking the average, then,
\[\mathbb{P}\left\{Y_{n+1}\in\widehat{C}_n^*(X_{n+1})\right\} = \frac{1}{m+1}\sum_{i=1}^{m+1}\mathbb{P}\left\{Y_{n+i}\in\widehat{C}_n^*(X_{n+i})\right\} = \mathbb{E}\left[ \frac{1}{m+1}\sum_{i=1}^{m+1}\mathbf{1}\left\{Y_{n+i}\in\widehat{C}_n^*(X_{n+i})\right\} \right].\]
By the tower law we can write this as
\[\mathbb{P}\left\{Y_{n+1}\in\widehat{C}_n^*(X_{n+1})\right\} =  \mathbb{E}\left[\mathbb{E}\left[ \frac{1}{m+1}\sum_{i=1}^{m+1}\mathbf{1}\left\{Y_{n+i}\in\widehat{C}_n^*(X_{n+i})\right\} \,\middle|\, X^0_{n+1},\dots,X^0_{n+m+1}\right]\right].\]
Equivalently, sampling an index $I$ uniformly from $\{n+1,\dots,n+m+1\}$,
\[\mathbb{P}\left\{Y_{n+1}\in\widehat{C}_n^*(X_{n+1})\right\} =  \mathbb{E}\left[\mathbb{P}\left\{Y_{n+I}\in\widehat{C}_n^*(X_{n+I}) \,\middle|\, X^0_{n+1},\dots,X^0_{n+m+1}\right\}\right].\]
But conditional on the group assignments
$X^0_{n+1},\dots,X^0_{n+m+1}$, we can observe that $(X_{n+I},Y_{n+I})$ is a draw from the distribution $\tilde{Q}^*$. Applying Theorem~\ref{thm:GWCP_cov_shift} (with $\tilde{Q}^*$ in place of $Q$), then,
\[\mathbb{P}\left\{Y_{n+I}\in\widehat{C}_n^*(X_{n+I}) \,\middle|\, X^0_{n+1},\dots,X^0_{n+m+1}\right\}\geq 1 - \alpha' - \max_{k:n_k>0} \{\tilde{q}^*_k/n_k\}\]
and so after marginalizing, 
\[\mathbb{P}\left\{Y_{n+1}\in\widehat{C}_n^*(X_{n+1})\right\} \geq 1-\alpha' - \mathbb{E}\left[\max_{k:n_k>0} \{\tilde{q}^*_k/n_k\} \right].\]
Finally, we have 
\[\mathbb{E}\left[(\tilde{q}_1,\dots,\tilde{q}_K)\mid (\tilde{q}^*_1,\dots,\tilde{q}^*_K),(n_1,\dots,n_K)\right] = (\tilde{q}^*_1,\dots,\tilde{q}^*_K),\]
so by Jensen's inequality, 
\[\mathbb{E}\left[\max_{k:n_k>0} \{\tilde{q}_k/n_k\} \,\middle|\, (\tilde{q}^*_1,\dots,\tilde{q}^*_K),(n_1,\dots,n_K)\right] \geq  \max_{k:n_k>0} \{\tilde{q}^*_k/n_k\}.\]
This completes the proof.
\end{proof}

\subsection{Is the coverage guarantee tight?}\label{app:tight_thm:GWCP_fixed_nk}
In this section, we examine whether the coverage guarantees obtained in this paper are tight, and whether less conservative versions of the 
method could be defined. 

We will first examine the question of unobserved groups. Note that the GWCP method defined in~\eqref{eqn:GWCP}
will compute a quantile $\widehat{q}=+\infty$, and consequently an infinite prediction interval $\widehat{C}(X_{n+1}) = \mathcal{Y}$,
if the proportion of unobserved groups is too large, i.e., if $\sum_{k:n_k=0} q_k >\alpha$.
The following proposition shows that this is unavoidable.
\begin{proposition}\label{prop:infinite_if_many_unobs} Let $\mathcal{Y}=\mathbb{R}$.
Let $\widehat{C}_n$ satisfy the distribution-free coverage guarantee
\[\mathbb{P}\{Y_{n+1}\in\widehat{C}_n(X_{n+1})\}\geq 1-\alpha\]
with respect to data $\{(X_i,Y_i)\}_{i=1,\dots,n}$ drawn as in~\eqref{eqn:distrib_fixed_nk} and test point $(X_{n+1},Y_{n+1})$
drawn as in~\eqref{eqn:define_Q_hier}, for \emph{any} distributions $\Pi_1,\dots,\Pi_K$.
If $\sum_{k:n_k=0}q_k>\alpha$ then
\[\mathbb{E}[\textnormal{Leb}(\widehat{C}_n(X_{n+1}))] = \infty.\]
\end{proposition}
\begin{proof}[Proof of Proposition~\ref{prop:infinite_if_many_unobs}]
Fix any $y\in\mathbb{R}$. We begin by defining an alternative distribution of the data: for each group $k$ let
\[\tilde{\Pi}_k = \begin{cases}\Pi_k, & n_k >0, \\
(\Pi_k)_{X^1}\times \delta_y, & n_k=0,\end{cases}\]
where $(\Pi_k)_{X^1}$ denotes the marginal of $X^1$ for $(X^1,Y)\sim\Pi_k$.
That is, if $n_k=0$, then the value $X^1$ is drawn exactly as before, while $Y$ is set to be equal to the constant value $y$.

Now let $(\tilde{X}_i,\tilde{Y}_i)$ denote data drawn with these new distributions $\tilde{\Pi}_k$ in place of the original $\Pi_k$'s,
and let $\tilde{C}_n$ denote the prediction interval when constructed on data drawn from this new distribution. Note that 
the training data and test feature,
$(\tilde{X}_1,\tilde{Y}_1),\dots,(\tilde{X}_n,\tilde{Y}_n),\tilde{X}_{n+1}$,
has an identical joint distribution as the original data,
$(X_1,Y_1),\dots,(X_n,Y_n),X_{n+1}$,
by construction, and therefore, 
\[\tilde{C}_n(\tilde{X}_{n+1})\stackrel{\textnormal{d}}{=}\widehat{C}_n(X_{n+1}),\]
where $\stackrel{\textnormal{d}}{=}$ denotes equality in distribution. Therefore, 
\[\mathbb{P}\{y\in \widehat{C}_n(X_{n+1})\} = 
\mathbb{P}\{y\in \tilde{C}_n(\tilde{X}_{n+1})\}.\]
Moreover, if $\tilde{X}_{n+1}^0 = k$ for some $k$ with $n_k=0$, then $\tilde{Y}_{n+1}=y$ almost surely, by definition of $\tilde{\Pi}_k$. Therefore,
\begin{multline*}\mathbb{P}\{y\in \tilde{C}_n(\tilde{X}_{n+1})\}
\geq \mathbb{P}\{\tilde{Y}_{n+1} \in \tilde{C}_n(\tilde{X}_{n+1}), \tilde{Y}_{n+1}=y\}\\
\geq \mathbb{P}\{\tilde{Y}_{n+1} \in \tilde{C}_n(\tilde{X}_{n+1}), \tilde{X}_{n+1}^0=k\textnormal{ for some $k$ with $n_k=0$}\}\\
\geq\mathbb{P}\{\tilde{X}_{n+1}^0=k\textnormal{ for some $k$ with $n_k=0$}\} -  \mathbb{P}\{\tilde{Y}_{n+1} \not\in \tilde{C}_n(\tilde{X}_{n+1})\} .\end{multline*}
And, $ \mathbb{P}\{\tilde{Y}_{n+1} \not\in \tilde{C}_n(\tilde{X}_{n+1})\}\leq \alpha$ (since the prediction interval satisfies
the distribution-free coverage guarantee---and in particular, must have coverage at least $1-\alpha$ with respect to this newly defined
distribution), while $ \mathbb{P}\{\tilde{X}_{n+1}^0=k\textnormal{ for some $k$ with $n_k=0$}\} = \sum_{k: n_k=0} q_k$. So, we have
\[\mathbb{P}\{y\in \tilde{C}_n(\tilde{X}_{n+1})\} \geq \sum_{k: n_k=0} q_k - \alpha >0.\]

Finally, the claim that $\mathbb{E}[\textnormal{Leb}(\widehat{C}_n(X_{n+1}))] = \infty$
follows directly from integration. For any $k$ with $n_k=0$,
\[\mathbb{E}[\textnormal{Leb}(\widehat{C}_n(X_{n+1}))]  \geq q_k \cdot \mathbb{E}[\textnormal{Leb}(\widehat{C}_n(X_{n+1})) \mid X_{n+1}^0=k]
 =  q_k \cdot \mathbb{E}\left[\mathbb{E}[\textnormal{Leb}(\widehat{C}_n(X_{n+1})) \mid X_{n+1}^0]\mid X_{n+1}^0=k\right],\]
and so it suffices to show that $\mathbb{E}[\textnormal{Leb}(\widehat{C}_n(x)) ]=\infty$ for any $x$ with $x^0=k$. We have
\begin{multline*}
\mathbb{E}[\textnormal{Leb}(\widehat{C}_n(x)) ]
=\mathbb{E}\left[\int_{y\in\mathbb{R}}\mathbf{1}_{y\in\widehat{C}_n(x)}\;\mathsf{d}y\right]
=\int_{y\in\mathbb{R}}\mathbb{E}\left[\mathbf{1}_{y\in\widehat{C}_n(x)}\right]\;\mathsf{d}y\\
=\int_{y\in\mathbb{R}}\mathbb{P}\{y\in\widehat{C}_n(x)\}\;\mathsf{d}y
\geq \int_{y\in\mathbb{R}} (\sum_{k: n_k=0} q_k - \alpha)\;\mathsf{d}y = \infty,
\end{multline*}
where the second step holds by the Fubini--Tonelli theorem, and the last step holds since $\sum_{k: n_k=0} q_k - \alpha>0$ by assumption.
\end{proof}

Next we consider the term $\max_k \{q_k/n_k\}$ in our coverage guarantee, Theorem~\ref{thm:GWCP_fixed_nk}, for the case of fixed group sizes $n_k$.
Can this error term be improved? The following counterexample establishes that this term is actually essentially optimal.
\begin{proposition}\label{prop:example_tight}
Let $\mathcal{Y}=\mathbb{R}$. Let $(1-\alpha)K$ be an integer, let $q_k \equiv 1/K$, and fix any $n_1,\dots,n_K\geq 1$. Without loss of generality assume $n_1 = \min_k n_k$.
Define a score function $s(x,y) =y$, and let $\Pi_k$ have marginal distribution $(\Pi_k)_Y$ on the response $Y$ given by
\[Y\sim \begin{cases}\textnormal{Unif}[0,1], & k=1,\\
\textnormal{Unif}[-1,0], & k=2,\dots,(1-\alpha)K,\\
\textnormal{Unif}[1,2],& k=(1-\alpha)K+1,\dots,K.\end{cases}\]
Let $q_k\equiv 1/K$.
Then the GWCP prediction interval~\eqref{eqn:GWCP} has coverage
\[\mathbb{P}\{Y_{n+1}\in\widehat{C}(X_{n+1})\} = 1-\alpha - \frac{1}{K(\min_k n_k+1)}.\]
\end{proposition}
Since $q_k\equiv 1/K$ in this example, the excess loss of coverage is therefore equal to $\frac{1}{K(\min_k n_k+1)} = \max_k \{q_k/(n_k+1)\}$, which (up to at most a factor of 2) matches the error term  $\max_{k:n_k>0} \{q_k/n_k\}$,
appearing in Theorem~\ref{thm:GWCP_fixed_nk}. Therefore, this establishes that the guarantee of this theorem is fairly tight.

\begin{proof}[Proof of Proposition~\ref{prop:example_tight}]
Since $q_k\equiv 1/K$, the weighted empirical distribution $\widehat{P}_{\textnormal{score}}$ is simply the average of the $\widehat{P}^{(k)}$'s.
Since we have scores lying in $[-1,0]$ for $(1-\alpha)K-1$ many groups, in $[0,1]$ for one group (namely, group $k=1$), and in $[1,2]$ for $\alpha K$ many groups, by definition of the quantile we can see that
$\widehat{q}$ is equal to the maximum score from group $k=1$---that is, $\widehat{q}$ is the maximum of $n_1$ many draws from the $\textnormal{Unif}[0,1]$ distribution. In particular, this implies that $\mathbb{E}[\widehat{q}] = \frac{n_1}{n_1+1}$ by properties of the uniform distribution.

Next, by construction of the score function, we have $\widehat{C}(X_{n+1}) = (-\infty,\widehat{q}]$, and so
\[\mathbb{P}\{Y_{n+1}\in\widehat{C}(X_{n+1})\} = \mathbb{P}\{Y_{n+1}\leq \widehat{q}\} = \sum_{k=1}^K q_k \mathbb{P}\{Y_{n+1}\leq \widehat{q}\mid X_{n+1}^0 = k\}.\]
Recall that $\widehat{q}$ must lie in $[0,1]$ by construction. Therefore,
if 
the test group is $X_{n+1}^0=k$ for some $k=2,\dots,(1-\alpha)K$, we have $Y_{n+1}\leq 0\leq \widehat{q}$ almost surely,
and if the test group is $X_{n+1}^0=k$ for some $k>(1-\alpha)K$, we have $Y_{n+1}>1\geq \widehat{q}$ almost surely.
On the other hand if $X_{n+1}^0=1$ then $Y_{n+1}$ is drawn from the $\textnormal{Unif}[0,1]$ distribution, and so 
\[\mathbb{P}\{Y_{n+1}\leq\widehat{q}\mid X_{n+1}^0=1\} = \mathbb{E}\left[\mathbb{P}\{Y_{n+1}\leq \widehat{q}\mid \widehat{q};X_{n+1}^0 =1\}\mid X_{n+1}^0 =1\right] = \mathbb{E}[\widehat{q}] = \frac{n_1}{n_1+1}.\]
We thus have
\[\mathbb{P}\{Y_{n+1}\in\widehat{C}(X_{n+1})\} = \frac{(1-\alpha)K - 1}{K} + \frac{1}{K} \cdot \frac{n_1}{n_1+1} = 1-\alpha - \frac{1}{K(n_1+1)}.\]
\end{proof}

\subsection{Which data should be used to estimate the weights?}\label{app:wcp_which_data set}
Recall that in Section~\ref{sec:compare_lower_bds}, we showed an empirical comparison between our new bounds and the coverage guarantees established 
by  \citep{lei2021conformal}'s existing result~\eqref{eqn:L1_w_bound}.
In addition to the more rapid convergence of the lower bound in our new result, there is another interesting
difference as compared to existing work. Specifically, existing theory (including \citep{lei2021conformal}'s result, but also related works such as \citep{candes2023conformalized, gui2022conformalized,yin2022conformal})
applies to an estimated weight function  $\widehat{w}(x)$ that was defined \emph{independently} of the calibration set.
For example, if a pretraining data set was used for defining the score function (e.g., fitting a model $\widehat{f}:\mathcal{X}\rightarrow\mathbb{R}$,
so that the score function is given by the residual score, $s(x,y) = |y - \widehat{f}(x)|$),
the weight function $\widehat{w}(x) $ might also be trained ahead of time by estimating the covariate distribution $P_X$ using this pretraining
data set. In contrast, the weights used in our proposed construction for the  GWCP method specifically use the \emph{calibration} set's proportions---recall from Section~\ref{sec:define_GWCP} that we can think of GWCP as essentially running WCP with weight function $\widehat{w}(X) = q_k / (n_k/n)$ for $X^0=k$, where $n_k$
 is the number of \emph{calibration} data points in group $k$.

For both types of results, the particular choice is essential to the theory:  \citep{lei2021conformal}'s result relies on $\widehat{w}$ being independent
of the calibration set, while in contrast, our proofs rely on using the $n_k$'s observed in the calibration set. Which approach is the ``right'' one for
achieving the best performance---that is, for the least loss of coverage?

To study this, we perform a simulation that compares the different options. We consider prediction intervals of the form
\begin{equation}\label{eqn:WCP_w_hat_k}
\widehat{q} = \textnormal{Quantile}_{1-\alpha}\left(\sum_{i=1}^n \frac{\widehat{w}(X_i)}{\sum_{i'=1}^{n}\widehat{w}(X_{i'})}\cdot \delta_{s_i}\right)\textnormal{ where }\widehat{w}(X_i) = \widehat{w}_k\textnormal{ for }X_i^0 =k,
\end{equation}
i.e., the weight $\widehat{w}(X_i)$ placed on data point $i$ depends only on the group $k$ to which this data point belongs. (As discussed above, in Section~\ref{sec:define_GWCP}, we can view this as a slightly less conservative version of the WCP method.) We will test three different options for defining the weights:
\begin{itemize}
    \item Estimate weights using the pretraining data set, $\{(X^*_i,Y^*_i)\}_{i\in[n_*]}$:\footnote{Since the estimated weight function $\widehat{w}$ is not well-defined on the rare event that any group has zero observations in the pretraining set, we discard any trials for which this occurs.}
    \[\widehat{w}_k = q_k/\widehat{p}_k \textnormal{ where }\widehat{p}_k = \frac{\sum_{i=1}^{n_*}\mathbf{1}_{X_i^*{}^0=k}}{n_*}.\]
    \item Estimate weights using the calibration data set, $\{(X_i,Y_i)\}_{i\in[n]}$:
    \[\widehat{w}_k = q_k/\widehat{p}_k \textnormal{ where }\widehat{p}_k = \frac{\sum_{i=1}^n\mathbf{1}_{X_i^0=k}}{n}.\]
    \item Use the oracle weights,
    \[\widehat{w}_k = q_k/p_k,\]
    where $p_k$ is the true probability of group $k$ under the training distribution $P$.
\end{itemize}
The distribution of the data is defined as follows: we have $K=5$ groups, with $q_k \equiv 0.2$. We set $n=100$ and $n_*=100$. The data distribution is given by $(X,Y)\in\{1,2,3,4,5\}\times \mathbb{R}$, where $Y\mid X=k \sim \mathcal{N}(\theta_k,1)$. 
The score function is fixed as $s(x,y)=|y|$.

We consider three settings for the group probabilities $p_k$ in the training distribution $P$, and the means $\theta_k$:
\begin{itemize}
    \item Setting 1: the groups are equally represented in the training data, with $(p_1,p_2,p_3,p_4,p_5) = (0.2,0.2,0.2,0.2,0.2)$, and $(\theta_1,\theta_2,\theta_3,\theta_4,\theta_5) = (20,15,10,5,0)$. 
    \item Setting 2: the group distributions are nonuniform, with $(p_1,p_2,p_3,p_4,p_5) = (0.4,0.25,0.2,0.1,0.05)$, and $(\theta_1,\theta_2,\theta_3,\theta_4,\theta_5) = (20,15,10,5,0)$. Note that in this setting, the more common groups tend to have higher values of the score.
    \item Setting 3: the group distributions are nonuniform, with $(p_1,p_2,p_3,p_4,p_5) = (0.4,0.25,0.2,0.1,0.05)$, $(\theta_1,\theta_2,\theta_3,\theta_4,\theta_5) = (0,5,10,15,20)$. Note that in this setting, the more common groups tend to have lower values of the score.
\end{itemize}

The results are displayed in Figure~\ref{fig:wcp_which_data set}.
We can see that in each setting, using the pretraining data for estimating the weights tends to result in undercoverage, while using the calibration data does not. (The slight undercoverage of the oracle method, in Setting 3, is due to the fact that we are not adding the correction term to our weighted conformal interval---that is, the difference between~\eqref{eqn:GWCP_alt_notation} and~\eqref{eqn:GWCP_plus}.) 

\begin{figure}[t]
     \centering
         \includegraphics[width = \textwidth]{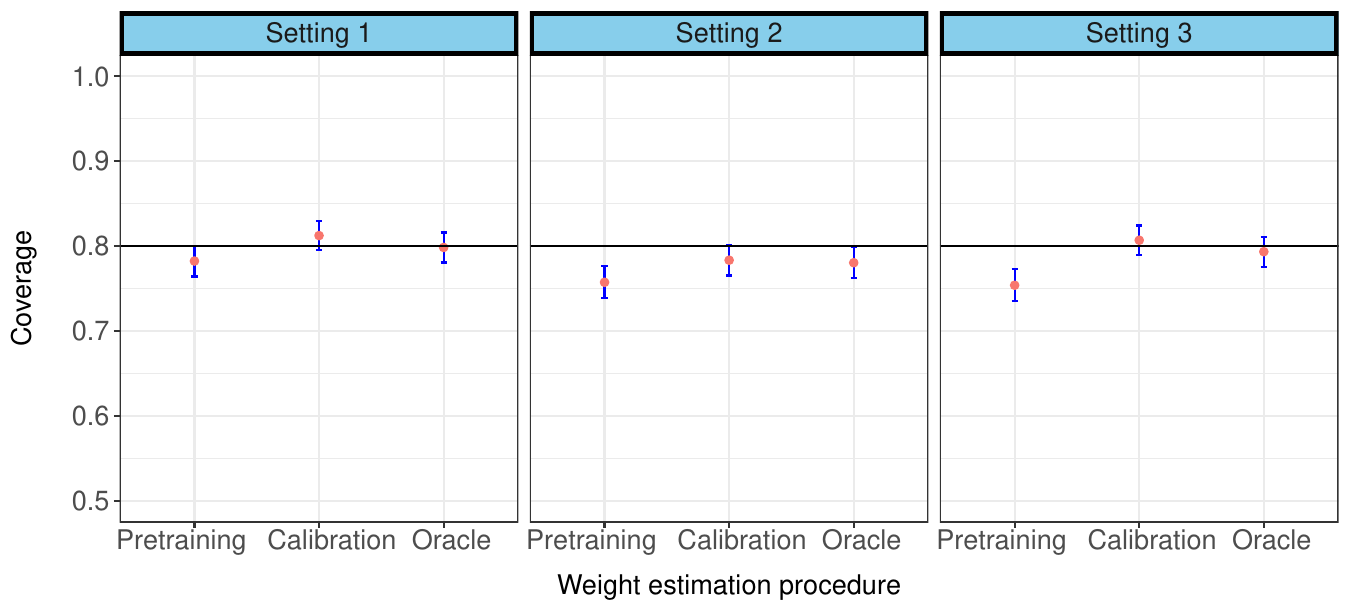}     
        \caption{Marginal coverage of the weighted conformal prediction interval~\eqref{eqn:WCP_w_hat_k}, implemented with weight function $\widehat{w}$ estimated using the pretraining data or using the calibration data, or, using the the ``oracle'' weights that rely on knowledge of the true distribution.
        The horizontal line marks the target coverage level $1-\alpha = 80\%$.
        Results are averaged over $2000$ independent trials, and the figure displays the mean coverage rate along with a 95\% confidence interval. See Appendix~\ref{app:wcp_which_data set} for details.}
        \label{fig:wcp_which_data set}
\end{figure}

Why does using the calibration data for estimating weights, lead to better coverage? Since we have chosen $n=n_*=100$ (i.e., $\widehat{w}$ is estimated using $100$ data points, for both options), this cannot be explained by a difference in the accuracy of estimating the true proportions $p_k$. Instead, we can understand this difference in performance as follows. When $\widehat{w}$ is estimated using the calibration data, the randomness in this estimate is \emph{self-correcting}: for instance, if for a particular group $k$ we have a slightly larger $n_k$ by random chance (i.e., $n_k > n p_k$), then
data points from this group are overrepresented when calculating the quantile $\widehat{q}$. But, the larger value of $n_k$ means that the weight $\widehat{w}_k$ is slightly smaller, and so these data points are also downweighted. In other words, there is a sort of cancellation effect, where the errors in the $\widehat{w}_k$'s are actually helpful for achieving the desired coverage level. In contrast, if we use the pretraining set to estimate the weights, the noise in the $n_k$'s is independent from the noise in the $\widehat{w}_k$'s, and so there is no such cancellation effect.

To summarize, in this work we have used a weight function $\widehat{w}$ trained on the calibration set
rather than the pretraining set, which is unusual relative to much of the literature---but this simulation demonstrates that this choice has better performance, as is supported by our stronger theoretical guarantees.

\subsection{A corrected GWCP for exact coverage}\label{app:GWCP_exact_coverage}
In this section, we provide a corrected version of the GWCP prediction interval that is guaranteed to provide coverage
at level $1-\alpha$. The goal is to construct a minimally more conservative version of GWCP to remove the error term in the coverage
guarantee.
The corrected interval is given by
\begin{equation}\label{eqn:GWCP_exact_coverage}
\widehat{C}_n(x) = \left\{y\in\mathcal{Y}: s(x,y) \leq \widehat{q}_k\right\}\textnormal{ for any $x\in\mathcal{X}$ with $x^0=k$},
\end{equation}
where $\widehat{q}_k = \textnormal{Quantile}_{1-\alpha_k}\left(\sum_{k=1}^K q_k \widehat{P}_{\textnormal{score}}^{(k)}\right)$
for
\[\alpha_k =\begin{cases} \alpha - \frac{q_k}{n_k}, & n_k >0, \\ \alpha, & n_k =0.\end{cases}\]
In other words, at a test point $X_{n+1}$ in group $k$, if $n_k>0$ then the threshold for the score $s(X_{n+1},y)$
is adjusted to be a bit more conservative; the adjustment is large only for groups with small observed sample size $n_k$.
(In the degenerate case that $\alpha_k<0$ for some $k$, we should interpret this by returning $\widehat{q}_k = +\infty$.)

\begin{theorem}\label{thm:GWCP_exact_coverage}
Suppose the training data points $(X_1,Y_1),\dots,(X_n,Y_n)$ are sampled from the fixed-group-size model, as in~\eqref{eqn:distrib_fixed_nk},
while the test point $(X_{n+1},Y_{n+1})$ is drawn independently from the distribution $Q$ as defined in~\eqref{eqn:define_Q_hier}.
Then, for any fixed (i.e., pretrained) score function $s:\mathcal{X}\times\mathcal{Y}\rightarrow\mathbb{R}$, the corrected GWCP prediction interval $\widehat{C}_n$ defined in~\eqref{eqn:GWCP_exact_coverage} satisfies
    \[\mathbb{P}\left\{Y_{n+1}\in\widehat{C}_n(X_{n+1})\right\} \geq 1 - \alpha.\]
\end{theorem}
\begin{proof}[Proof of Theorem~\ref{thm:GWCP_exact_coverage}]
Following the same notation and terminology as the proof of Theorem~\ref{thm:GWCP_fixed_nk},
we rewrite the coverage probability as
    \[\mathbb{P}\left\{Y_{n+1}\in\widehat{C}_n(X_{n+1})\right\} = \sum_{k=1}^K q_k\mathbb{P}\left\{s(X_{k,0},Y_{k,0}) \leq \textnormal{Quantile}_{1-\alpha_k}(\widehat{P}_{\textnormal{score}}^{k,0})\right\},
   \]
   i.e., this is the same calculation as in~\eqref{eqn:key_step_thm:GWCP_exact_coverage} except with $\alpha_k$ in place of $\alpha$ for each group $k$.
The next step is again the same: by symmetry, we can rewrite this as
    \[\mathbb{P}\left\{Y_{n+1}\in\widehat{C}_n(X_{n+1})\right\} = \sum_{k=1}^K\sum_{i=1}^{n_k} \frac{q_k}{n_k}\mathbb{P}\left\{s(X_{k,i},Y_{k,i}) \leq \textnormal{Quantile}_{1-\alpha_k}(\widehat{P}_{\textnormal{score}}^{k,i})\right\}.\]
Now consider any group $k$. If $n_k=0$ then $\widehat{P}_{\textnormal{score}}^{k,i} = \widehat{P}_{\textnormal{score}}$ and $\alpha_k =\alpha$, and thus,
$\textnormal{Quantile}_{1-\alpha_k}(\widehat{P}_{\textnormal{score}}^{k,i}) = \textnormal{Quantile}_{1-\alpha}(\widehat{P}_{\textnormal{score}}) $. If instead $n_k>0$, then
by construction, for any $i$ we have $\textnormal{d}_{\textnormal{TV}}(\widehat{P}_{\textnormal{score}},\widehat{P}_{\textnormal{score}}^{k,i})\leq  q_k/n_k = \alpha - \alpha_k$, and therefore,
$ \textnormal{Quantile}_{1-\alpha_k}(\widehat{P}_{\textnormal{score}}^{k,i}) \geq  \textnormal{Quantile}_{1-\alpha}(\widehat{P}_{\textnormal{score}}) $.
(This holds also for the case $\alpha_k<0$ since then the quantile on the left-hand side is defined as $+\infty$.)
Combining both cases, then,
    \[\mathbb{P}\left\{Y_{n+1}\in\widehat{C}_n(X_{n+1})\right\} \geq \sum_{k=1}^K\sum_{i=1}^{n_k} \frac{q_k}{n_k}\mathbb{P}\left\{s(X_{k,i},Y_{k,i}) \leq \textnormal{Quantile}_{1-\alpha}(\widehat{P}_{\textnormal{score}})\right\}\geq 1-\alpha,\]
where (as in the proof of Theorem~\ref{thm:GWCP_fixed_nk}) the last step holds by definition of the weighted empirical distribution $\widehat{P}_{\textnormal{score}}$ (see, e.g., \citet[Lemma A.1]{harrison2012conservative}).    
\end{proof}

To complete this section, we repeat the simulation in Section~\ref{sec:sim_fixed_nk}, with the corrected GWCP~\eqref{eqn:GWCP_exact_coverage} in place of the original version~\eqref{eqn:GWCP}. The simulation settings are exactly the same as in Section~\ref{sec:sim_fixed_nk}. 
The results are displayed in Figure~\ref{fig:sim_fixed_nk_corrected}. In this figure, coverage is always at least $1-\alpha$,
as guaranteed by Theorem~\ref{thm:GWCP_exact_coverage}.
In particular, comparing to Figure~\ref{fig:sim_fixed_nk}, we can see that the loss of coverage
incurred when $\min_k n_k=1$ is no longer present---but instead, the method is somewhat too conservative for the \texttt{AllSmall} setting, for smaller values of $K$.

\begin{figure}[t]
    \centering
    \includegraphics[width = 0.9\textwidth]{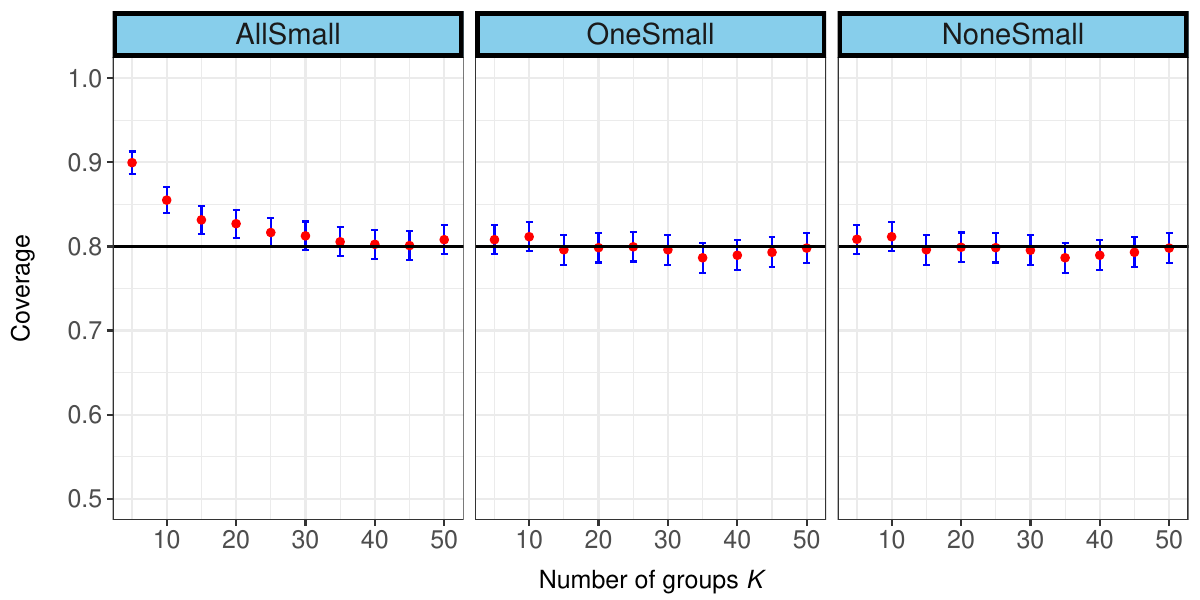}
    \caption{Simulation results for the setting of fixed group sizes, in three regimes for the group sizes $n_1,\dots,n_K$, with the corrected GWCP prediction interval~\eqref{eqn:GWCP_exact_coverage} in place of the original version~\eqref{eqn:GWCP}.
    The horizontal line marks the target coverage level $1-\alpha = 80\%$. Results are averaged over $2000$ independent trials, and the figure displays the mean coverage rate along with a 95\% confidence interval.  See Appendix~\ref{app:GWCP_exact_coverage} for details. {\color{red} }}
    \label{fig:sim_fixed_nk_corrected}
\end{figure}

\end{document}